\documentclass[12pt]{article}
\usepackage[margin=1in]{geometry}

\usepackage[affil-it]{authblk}
\usepackage{etoolbox}
\usepackage{lmodern}
\usepackage{rotating}
\makeatletter
\patchcmd{\@maketitle}{\LARGE \@title}{\fontsize{14}{17.2}\selectfont\@title}{}{}
\makeatother

\usepackage{sectsty}

\sectionfont{\fontsize{13}{15}\selectfont}
\subsectionfont{\fontsize{12}{15}\selectfont}

\usepackage{setspace}

\usepackage{natbib}

\usepackage{graphicx}
\usepackage{url} 

\newcommand{\blind}{1}

\usepackage[utf8]{inputenc} 
\usepackage[T1]{fontenc}    
\usepackage{hyperref}       
\usepackage{url}            
\usepackage{booktabs}       
\usepackage{amsfonts}       
\usepackage{nicefrac}       
\usepackage{microtype}      
\usepackage{xcolor}         

\usepackage{algorithm}
\usepackage{amsmath}
\usepackage{mathtools}
\usepackage{dsfont}
\usepackage{graphicx,psfrag,epsf} 
\usepackage{enumerate}
\usepackage{amssymb}
\usepackage{algpseudocode}

\usepackage{amsfonts}

\usepackage{subcaption}
\usepackage{multirow}
\usepackage{amsthm}
\usepackage{hyperref}
\usepackage{soul}
\usepackage{bm}
\usepackage{enumitem}

\usepackage[T1]{fontenc}
\usepackage{adjustbox}

\newtheorem{theorem}{Theorem}
\newtheorem{definition}[theorem]{Definition}

\newtheorem{remark}[theorem]{Remark}
\newtheorem{corollary}[theorem]{Corollary}
\newtheorem{assumption}{Assumption}

\newcommand{\GP}{$\mathcal{GP}$}

\begin{document}

\bibliographystyle{apalike}

\def\spacingset#1{\renewcommand{\baselinestretch}%
{#1}\small\normalsize} \spacingset{1}


\if1\blind
{
  \title{\bf The Traveling Bandit: A Framework for Bayesian Optimization with Movement Costs}
    \author[1]{Qiyuan Chen}
    \author[1, *]{Raed Al Kontar}
    \date{} 

    \affil[1]{University of Michigan, Ann Arbor}
    \affil[*]{Corresponding author: alkontar@umich.edu}
  \maketitle
} \fi

\if0\blind
{
    \title{\bf The Traveling Bandit: A Framework for Bayesian Optimization with Movement Costs}
  \author{}
  \date{} 
\maketitle
} \fi

\bigskip
\begin{abstract}
This paper introduces a framework for Bayesian Optimization (BO) with metric movement costs, addressing a critical challenge in practical applications where input alterations incur varying costs. Our approach is a convenient plug-in that seamlessly integrates with the existing literature on batched algorithms, where designs within batches are observed following the solution of a Traveling Salesman Problem. The proposed method provides a theoretical guarantee of convergence in terms of movement costs for BO. Empirically, our method effectively reduces average movement costs over time while maintaining comparable regret performance to conventional BO methods. This framework also shows promise for broader applications in various bandit settings with movement costs.
\end{abstract}

\noindent%

\medskip

\spacingset{1.2}

\section{Introduction}
\label{sec:intro}

Bayesian Optimization \citep[BO, ][]{gramacy2020surrogates}, is a class of global optimization methods for expensive-to-evaluate black-box functions. BO sequentially selects inputs, referred to as designs hereafter, to observe with the goal of finding an optimal solution using the fewest number of experiments. In a basic black-box optimization setup, the costs of experiments are treated equally, so minimizing experimental costs is equivalent to minimizing the number of experiments. However, in many practical applications, altering the designs to a black-box function can incur considerable costs, which we call a movement cost. In mineral exploration, moving drilling rigs from one location to another incurs both time and costs. Such movement is not restricted only to physical movements. In a 3D printing example, the time to change the temperature of a 3D printer's hotend is often considerable, especially when the hotend is raised or cooled to a very different temperature. Unlike the settings for cost-aware BO \citep{lee2021nonmyopic}, movement cost does not solely depend on the chosen design itself, but depends on the previous design chosen by the algorithm, making it more challenging than design-specific costs. Indeed, movement costs are prevalent in various real-world problems, including industrial processes \citep{teh2008droplet}, wind energy \citep{ramesh2022movement}, e-commerce \citep{koren2017bandits}, and environmental monitoring \citep{leu2020air,samaniego2021bayesian,hellan2022bayesian}.

To further highlight our objective, Fig. \ref{fig:Branin} provides a visual comparison of the sample paths (in red) taken by a movement-cost unaware method and the movement cost-aware method proposed in this paper. The function values are plotted in the background: the cooler the colors, the better the function values. Although both algorithms visit similar design points, it is evident that the movement cost-aware method significantly reduces movement by following a more efficient route through them.

\begin{figure}[ht]
    \centering
        \begin{subfigure}{0.49\linewidth}
        \centering
        \includegraphics[width=0.49\textwidth]{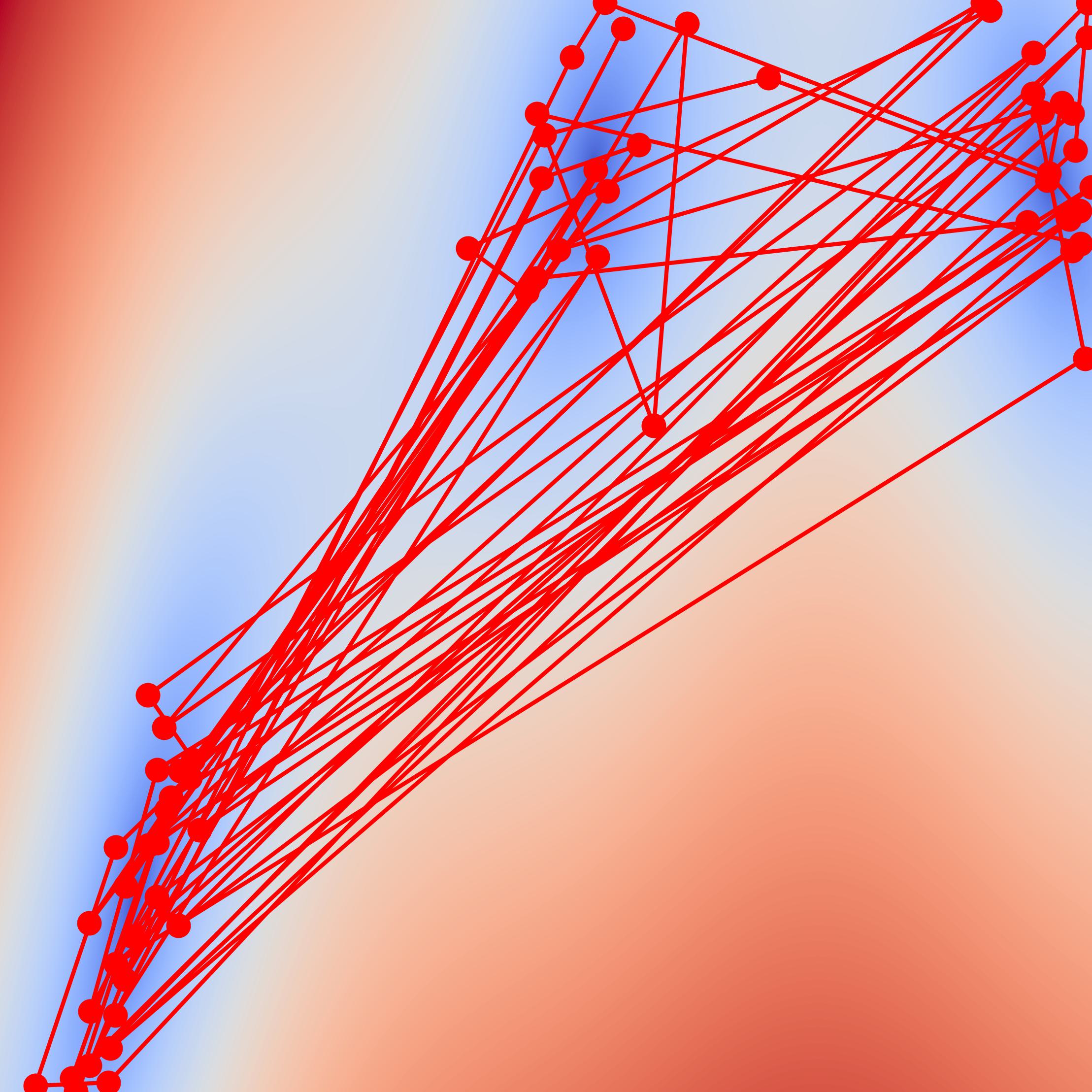}
        \caption{Cost-unaware}
        \label{fig:ts}
    \end{subfigure}
    \hfill
    \begin{subfigure}{0.49\linewidth}
        \centering
        \includegraphics[width=0.49\textwidth]{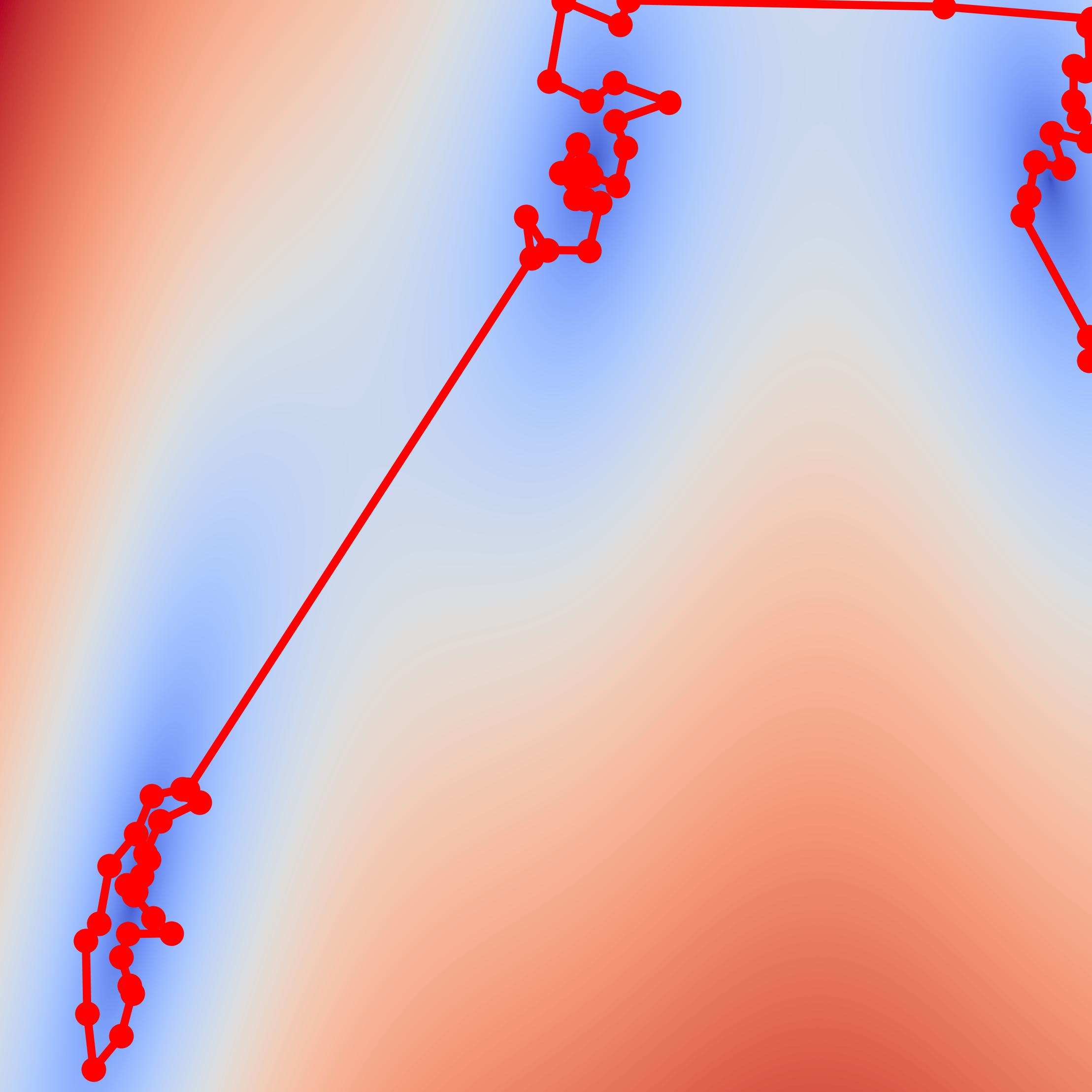}
        \caption{Cost-aware}
        \label{fig:tts}
    \end{subfigure}
    \caption{Sample path of  movement-cost unaware and aware approaches}
    \label{fig:Branin}
\end{figure}

The fundamental challenge in BO with movement costs lies in the conflict between exploration and movement. On the one hand, exploration of various designs is crucial for gauging good solutions. On the other hand, drastic changes in designs can lead to high movement costs. In fact, the need for movement is highly dependent on the need for exploration, which is determined by the hardness of the black-box function.  Fig. \ref{fig:quadvscon} demonstrates this difference with two extreme examples. The figures present the last half of the points visited by a movement cost-unaware method and our method. When there is an obvious unique optimal solution (as shown in the quadratic function in Fig.\ref{fig:ucb_quad} and \ref{fig:tpe_quad}), movement costs can converge automatically, even without being explicitly considered. This is because optimizing the function value naturally aligns with minimizing movement. Conversely, if the function has multiple comparable solutions (as shown in the constant function in Fig.\ref{fig:ucb_con} and \ref{fig:tpe_con}), any reasonable algorithm will continue to iteratively explore all of them, which results in significant movement costs.

\begin{figure}[htb]
    \centering
    \begin{subfigure}{0.24\linewidth}
        \centering
        \includegraphics[width=\textwidth]{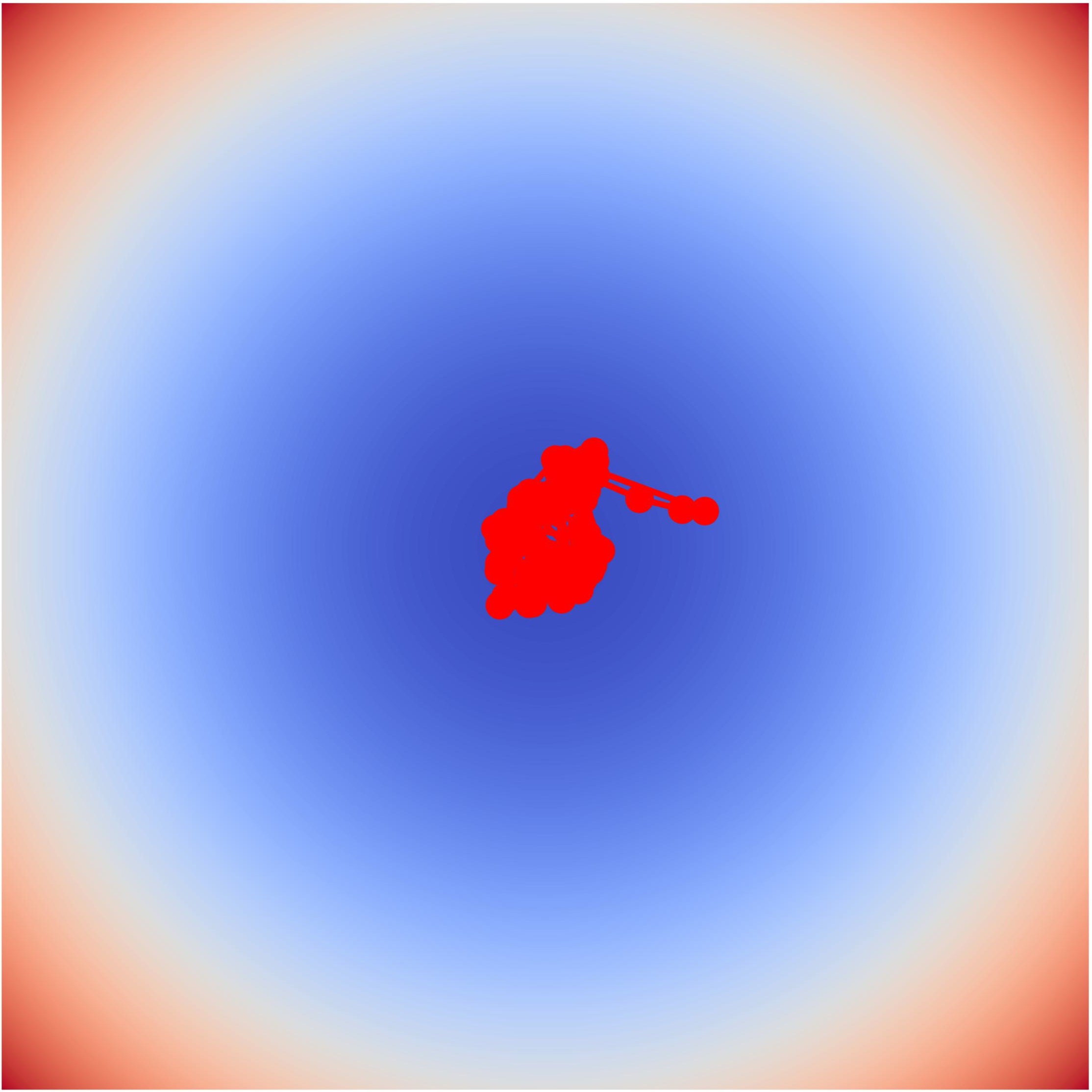}
        \caption{Cost-unaware}
        \label{fig:ucb_quad}
    \end{subfigure}
    \hfill
    \begin{subfigure}{0.24\linewidth}
        \centering
        \includegraphics[width=\textwidth]{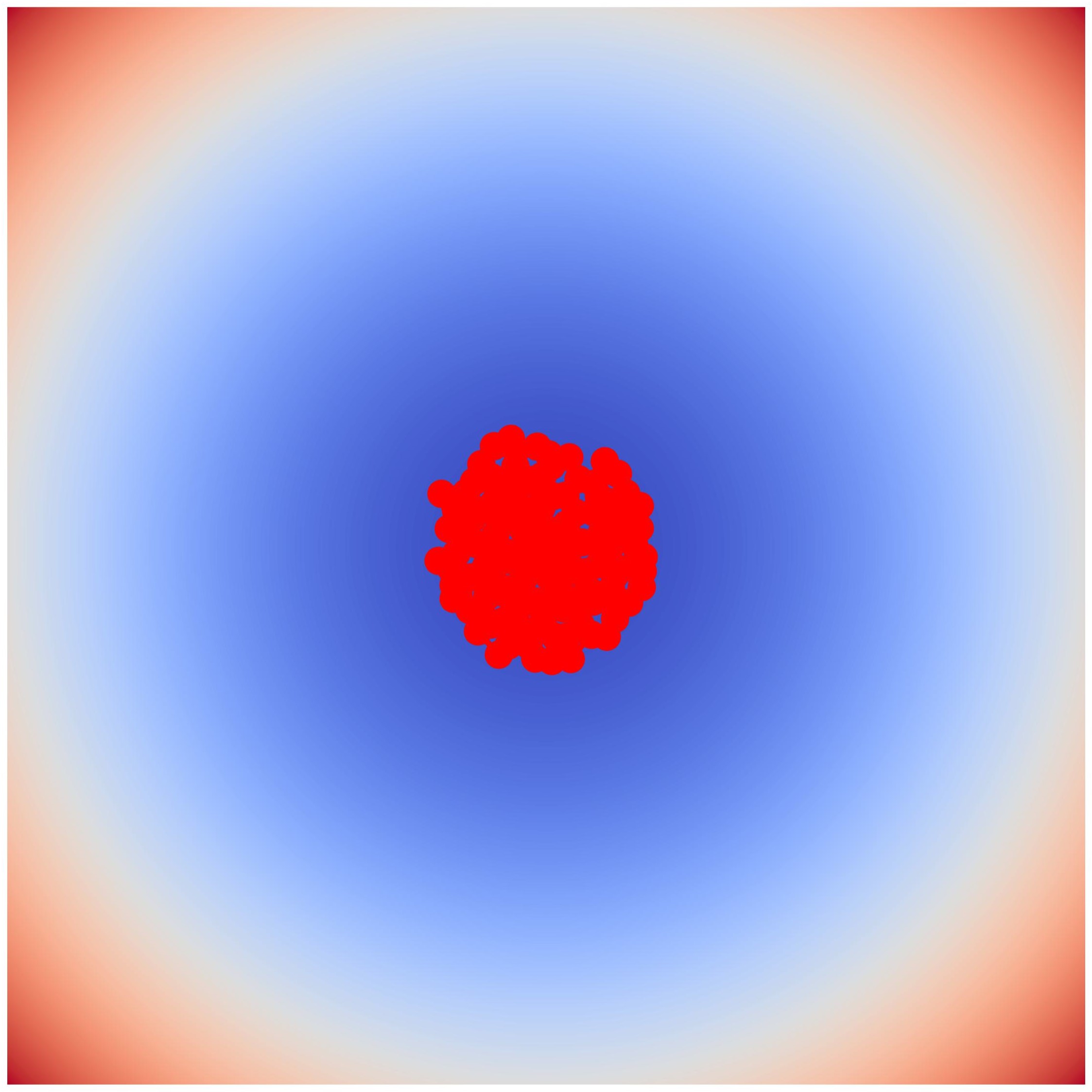}
        \caption{Cost-aware}
        \label{fig:tpe_quad}
    \end{subfigure}
    \hfill
    \begin{subfigure}{0.24\linewidth}
        \centering
        \includegraphics[width=\textwidth]{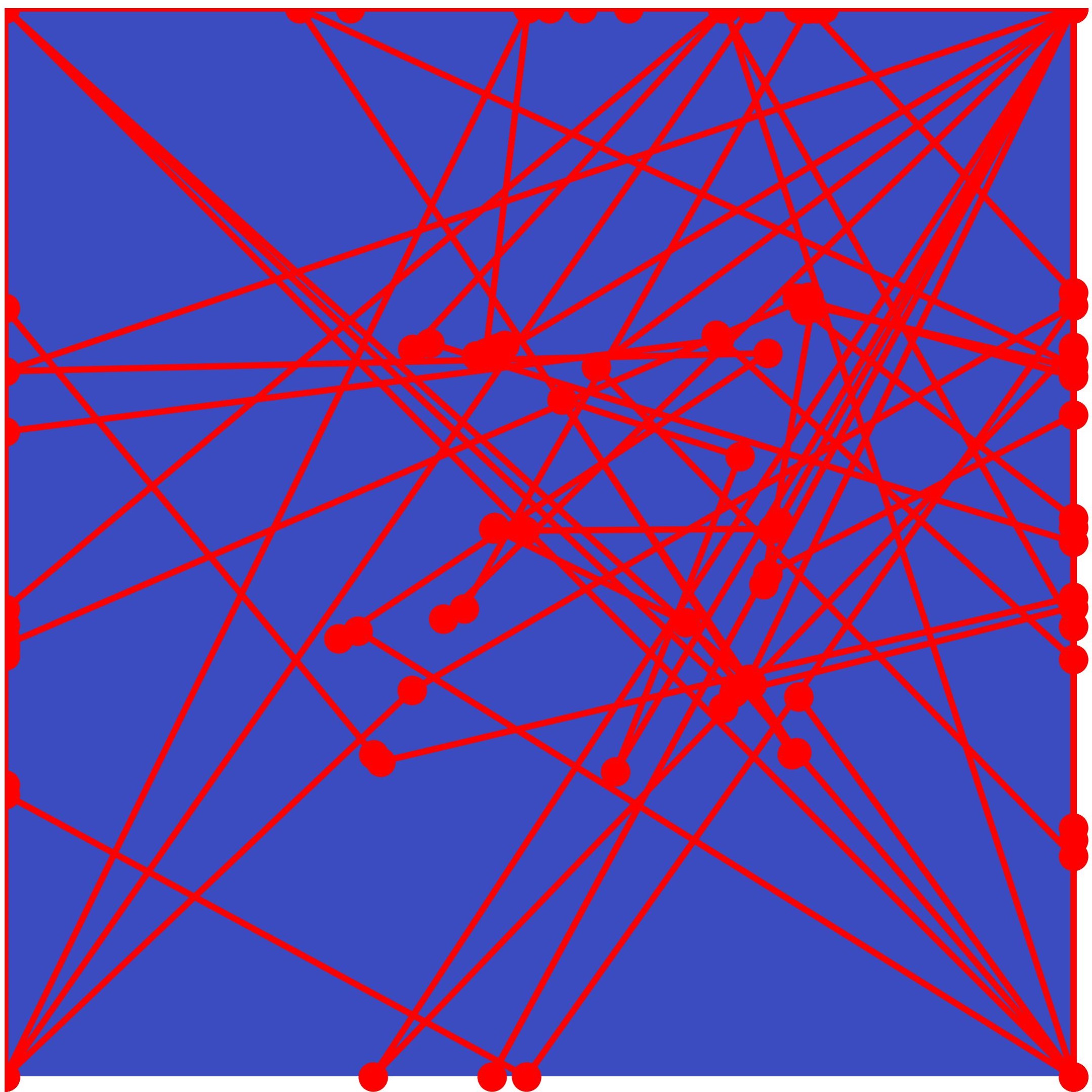}
        \caption{Cost-unaware}
        \label{fig:ucb_con}
    \end{subfigure}
    \hfill
    \begin{subfigure}{0.24\linewidth}
        \centering
        \includegraphics[width=\textwidth]{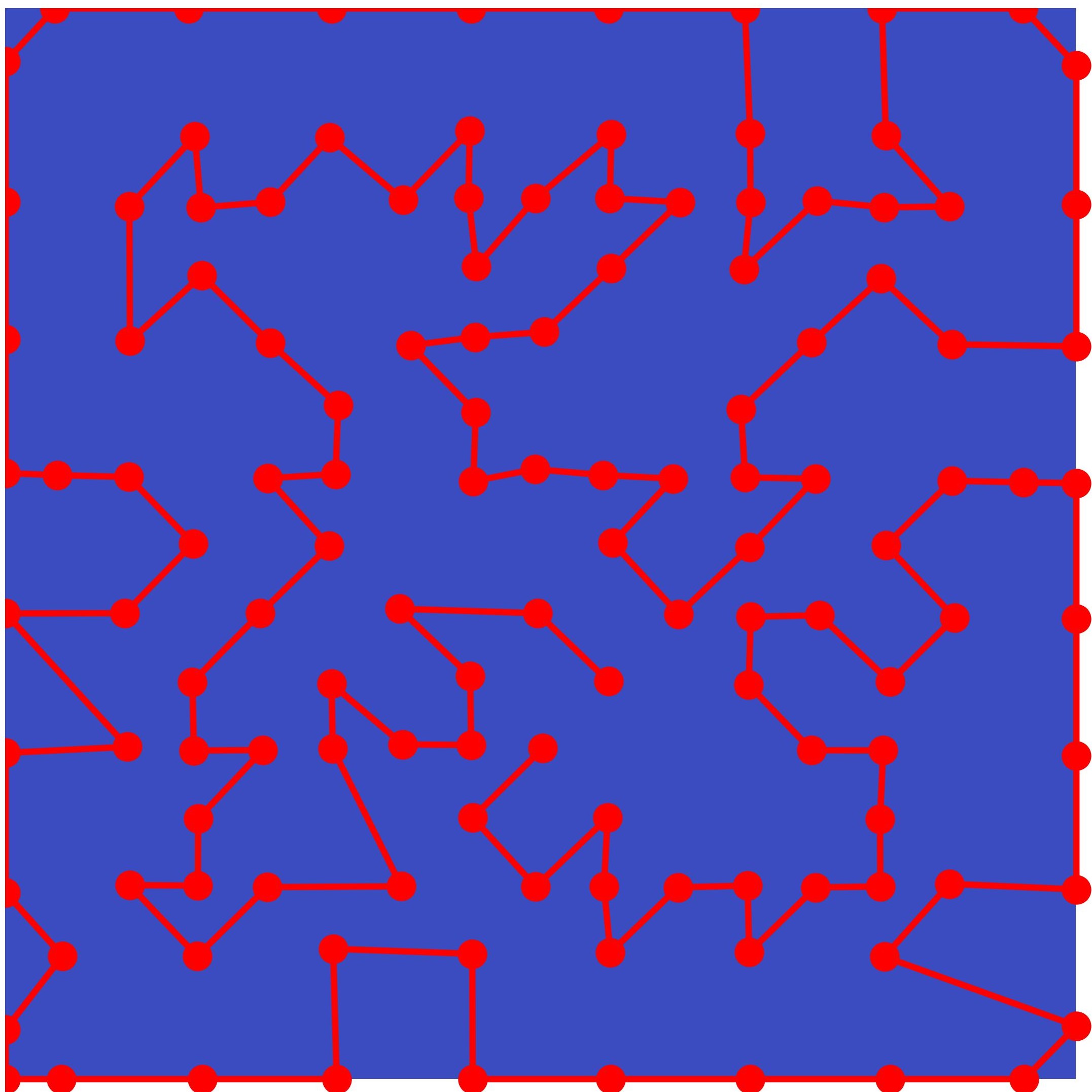}
        \caption{Cost-aware}
        \label{fig:tpe_con}
    \end{subfigure}
    \caption{Sample path of a movement cost-unaware vs. our cost-aware method.}
    \label{fig:quadvscon}
\end{figure} 

In the context of black-box optimization, the hardness of the problem is unknown in advance, so global optimization algorithms should always be prepared to handle a scenario where there are multiple, often comparable, local solutions.  Despite a few attempts to integrate movement costs into BO \citep{samaniego2021bayesian, folch2022snake, shi2023bayesian, liu2023bayesian}, these methods fail to be global optimization approaches due to their myopic and greedy nature. Specifically, existing methods penalize the acquisition function based on movement costs. Typical algorithms of this kind includes \texttt{EIpu} \citep{lee2021nonmyopic}, \texttt{GS} \citep{shi2023bayesian}, and \texttt{distUCB} \citep{liu2023bayesian}. Despite the simplicity of the idea, penalizing exploration with movement costs makes the algorithm overly myopic. When movement costs are high, such penalization significantly reduces exploration and ultimately prevents a method from functioning as a global optimization algorithm.

For instance, consider a function with two far-away solutions where one (in green) is better than the other (in red), shown in Fig.\ref{fig:penalty}. Suppose we start with the worse solution, so one should expect that any effective algorithm should eventually bring us to the other local optimal solution. However, for myopic methods, as long as the distance between the two local optimal solutions is far enough (i.e. large movement costs), exploration could be over-penalized so that the algorithm never visits the other solution.

\begin{figure}[htbp!]
    \centering
    \includegraphics[width=0.75\linewidth]{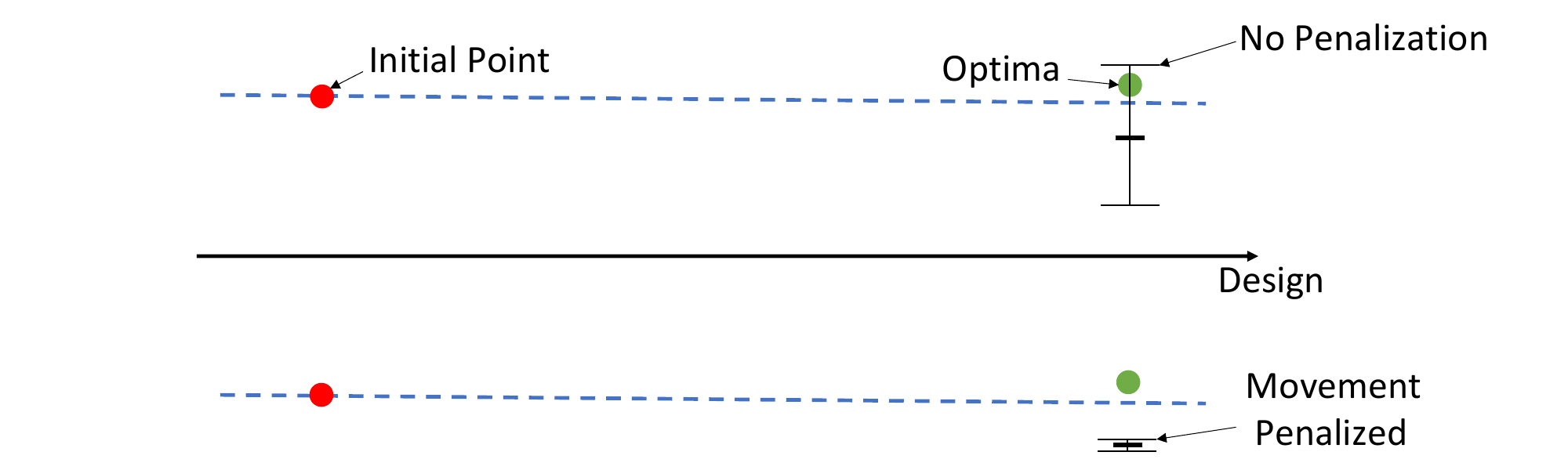}
    \caption{The Failure of Movement Penalization}
    \label{fig:penalty}
\end{figure}

The possibility of having multiple local solutions poses a fundamental challenge for BO with movement costs. Traveling back and forth through all the possible solutions can be very costly, and trading exploration to save movement can make the algorithm stuck at local optimal solutions and fail to approach global optimality. Instead of trying to balance exploration and movement, we take an alternative route by decoupling the optimization of the function (measured in terms of regret) and the movement costs. To this end, we introduce a universal framework for BO under movement costs that is both powerful and simple in design. It involves only two iterative steps. The first step focuses on optimizing the black-box function by selecting a set of designs that maximize a user-defined acquisition function. The second step focuses on minimizing movement costs by planning a route through the selected designs. 

The contribution of this paper is summarized as follows: (1) We present an easy-to-implement BO approach to save on movement costs. (2) We show that this simple framework enjoys elegant theoretical guarantees (Corollary \ref{thm:general_move_costs}), where the average movement costs approach zero asymptotically without compromising regret performance. (3) We validate the algorithm's practical effectiveness empirically across various acquisition and test functions, confirming its superior practicality compared to state-of-the-art methods. (4) Beyond BO, our method can be readily extended to many other online decision-making problems, such as stochastic bandit problems. We provide two examples of multi-armed and Lipschitz bandits. In both settings, the cumulative movement costs are asymptotically less than the cumulative regret, achieving \textit{optimal} rates for the cumulative loss (i.e., cumulative regret plus cumulative movement costs) in these two problems.

\section{Related Work}
\label{sec:lit}
\paragraph{BO with Movement Costs}
BO has been a popular tool for the global optimization of black-box functions since \citet{kushner1963new}. Recently, there has been an increasing interest in BO with movement costs \citep{samaniego2021bayesian,folch2022snake} as it factually models the working conditions in many real-world settings \citep{samaniego2021bayesian,hellan2022bayesian,ramesh2022movement}. Recent methods have been proposed to solve this problem, including \texttt{SnAKe} \citep{folch2022snake}, \texttt{GS} \citep{shi2023bayesian}, and \texttt{distUCB} \citep{liu2023bayesian}. Along this line, movement costs are sometimes also considered a constraint \citep{paulson2022efficient,paulson2023lsr} rather than a target to minimize. However, we note that our problem setting does not belong to a constrained BO setting \citep{bernardo2011optimization} since its constraints depend not only on the design but also on the history of its visited designs. 

The algorithm proposed in this paper is closest to the idea of \texttt{SnAKe} \citep{folch2022snake}, which also plans paths for the future. However, \texttt{SnAKe} uses the planned paths in a fundamentally different way. First, their path is planned for the rest of the entire horizon, which can be challenging when the horizon is infinite or not known in advance. Second, \texttt{SnAKe} does not finish the path they planned but switches to a new path when new observations are available. Consequently, \texttt{SnAKe} converges neither theoretically nor empirically. In contrast, this paper provides a convergence guarantee for the BO with movement costs. As will become clear shortly, our guarantees hinge on a simple yet elegant result for shortest-path planning.

We acknowledge that there is another line of BO research, although initially aimed at reducing the Gaussian process (\GP{}) estimation complexity, can also help reduce movement costs. Namely, the \texttt{MINI} algorithm \citep{calandriello2022scaling} is designed to learn a \GP{} by evaluating a few unique candidate designs multiple times when learning a \GP{}. As a by-effect, evaluating only few candidates automatically results in low movement costs. However, empirically, this results in a suboptimal estimation of the black-box function landscape and kernel hyperparameters. Also, \texttt{MINI} is specific to \GP{} surrogates only, while our method is more generic and suits a much broader class of black-box optimization problems. 

Although we focus on BO in this paper, our result has a broader impact and can be applied to generic online decision problems.  Two special cases of the problem described in Sec.\ref{sec:setup} are also studied under the umbrella of bandits, namely multi-armed (MAB) and Lipschitz bandits.  MAB with movement costs was first studied as switching costs \citep{cesa2013online,dekel2014bandits, rouyer2021algorithm, steiger2023constrained}, where the metric space is a connected finite graph with the distance between all pairs of vertices being $1$. Unlike our setting, this literature studies adversarial MAB \citep{cesa2013online,dekel2014bandits}.

For stochastic bandit problems, the analysis of switching costs can extended to movement costs \citep{guha2009multi, koren2017bandits, koren2017multi}, where the costs of input changes are defined by a known metric on the input space. Among them, \citet{koren2017multi} studied movement costs for Lipschitz bandits in a bounded metric space of $d-$Minkowski dimension and obtained a tight bound. This paper falls in this category. We follow the convention of \citet{koren2017multi} and analyze the movement costs using the Minkowski dimension. Interestingly, our algorithm matches the same bound for Lipschitz bandits with movement costs as \citet{koren2017multi}.

\paragraph{Batched Algorithms}
A key contribution of our work is pointing out the automatic convergence of movement costs if actions are planned ahead.  Such algorithms that plan \textit{multiple} future actions ahead are better known as batched algorithms. Batched algorithms are initially designed to parallelize experiments, where one needs to determine multiple designs to observe simultaneously without knowing their outcomes. Batched algorithms have long been an active research topic in BO \citep{azimi2010batch,duan2017sliced,hunt2020batch}. Although batched algorithms are initially designed for batched feedback, they can be readily applied to non-batched settings. Recently, \cite{li2022gaussian} used a batched algorithm to provide a first tight regret upper bound for non-batched settings. The success of these algorithms implies that planning decisions in batches can achieve regret bounds similar to traditional non-batched algorithms, which makes our method possible.

We note that despite the fact that we use batched algorithms, our problem is \textbf{not} limited to batched settings. Specifically, we do allow the algorithm to see the outcome of the tested design instantly before it can make another decision. The purpose of using a batched algorithm in a non-batched setting is only to find a set of good candidate designs so that we can plan a route to save movement. In the experiment section, all the algorithms are benchmarked in the environment where each single outcome is revealed immediately after testing. 

\paragraph{Traveling Salesman Problem}
Our algorithm requires solving a Traveling Salesman Problem (TSP) in its sub-routine. TSP involves finding the shortest path on a graph that goes through all the vertices of the graph. It is well-known that TSP is an NP-hard problem. However, when the distances follow triangular inequalities, the problem is known as metric-TSP and can be solved to $3/2$-optimal in polynomial time \citep{christofides1976worst}. This constant approximation is sufficient for our theoretical purposes. In practice, other heuristic algorithms \citep{lin1973effective} might offer better solution quality and shorter running times, so we leave the choice to practitioners

\section{Setting the Stage}
\subsection{Problem Setup}
\label{sec:setup}
Let $(M,c)$ be a metric space where set $M$ is endowed with a distance metric $c$. Consider a real-valued black-box function $f: M \rightarrow \mathbb{R}$. The goal of the entity running the experiments is to find the maximizer of the black-box function $f$ within a budget of $T$ experiments. Hereafter, we refer to this entity as the traveler. Let the $t^{\text{th}}$ destination and its associated observation be $\bm{x}_t$ and $y_t$, respectively. Denote the history as a dataset $\mathcal{D}_{t}\triangleq \{(\bm{x}_1,y_1),\dots,(\bm{x}_t,y_t)\}$.  Realistically, the observation $y_t$ can be noisy such that $y(\bm{x}_{t}) \triangleq f(\bm{x}_{t})+\varepsilon_{t}$, where $\varepsilon_{t}$ are R-sub-Gaussian random variables with $\mathbb{E}[\varepsilon_{t}|\mathcal{D}_{t-1}]=0$. 

In addition to a conventional setting, the traveler also needs to pay a movement cost $c(\bm{x}_{t-1},\bm{x}_{t})$ to travel from one location $\bm{x}_{t-1}$ to the next destination $\bm{x}_{t}$. Therefore, this is essentially a multiple-objective problem, where the traveler's goal is to apply a policy $\pi$ that not only maximizes the function values but also minimizes the movement costs.  

\subsection{BO with a Gaussian Process Surrogate} \label{sec:bo}

BO with a \GP{} surrogate, also known as \GP{} bandit optimization, is an optimization policy that selects designs at each round $t$ by optimizing an acquisition function $\alpha(\bm{x}_{t};\mathcal{D}_{t-1})$, which is calculated based on a \GP{}-based posterior belief of the underlying black-box function. For the most part of this paper, we limit the discussion and theoretical analysis to BO with a $\mathcal{GP}$ posterior, although the results can be readily adapted to other surrogates and even other stochastic problem settings (see other examples in Sec.\ref{sec:converge}). 

A $\mathcal{GP}$ is a nonparametric model defined by a positive-definite kernel $k(\cdot,\cdot)$ and the past observations $\mathcal{D}_{t}$. The kernel $k$ imposes a functional prior on the black box $f$, which is updated to a posterior belief by conditioning on the observed data $\mathcal{D}_{t}$. Now, for simplicity, we set up some shorthand notations. At time $t$, define $\bm{y}_{1:t}\triangleq [y_1,\dots,y_t]^\intercal$, $\bm{k}_{t}(\bm{x})= [k(\bm{x},\bm{x}_{1}),\cdots, \bm{k}(\bm{x},\bm{x}_{t})]^\intercal$ and the Gram matrix $\bm{K}_t=[\bm{k}_{t}(\bm{x}_{1}), \dots, \bm{k}_{t}(\bm{x}_{t})]$. By a $\mathcal{GP}(\mathcal{D}_{t})$ construction, we have that: 
\[\begin{bmatrix}
    \bm{y}_{1:t}\\
    f(\bm{x})
\end{bmatrix}
~\sim~
\mathcal{N}\left(\bm{0},
\begin{bmatrix}
    \bm{K}_t + \lambda I & \bm{k}_{t}(\bm{x})\\
    \bm{k}_{t}(\bm{x})^\intercal& k(\bm{x},\bm{x})
\end{bmatrix}\right)
,\]
where $\bm{y}_{1:t} = [y_{1},\dots,y_{t}]^\intercal$ and $\lambda$ is a nugget effect that provides regularization and alleviates numerical problems for $\mathcal{GP}$s \citep{matheron1963principles}. By marginalizing all the possible function realizations, we get a posterior belief $f(\bm{x})\mid \mathcal{D}_{t} \sim \mathcal{N}(\mu_{t}(\bm{x}), \sigma^2_{t}(\bm{x}))$, where 
\[\mu_{t}(\bm{x}) =  \bm{k}_t(\bm{x})^\intercal (\bm{K}_t+ \lambda \bm{I})^{-1} \bm{y}_{1:t},\] \[\sigma^2_{t}(\bm{x}) =  k(\bm{x},\bm{x}) - \bm{k}_t(\bm{x})^\intercal (\bm{K}_t+ \lambda I)^{-1}\bm{k}_t(\bm{x})\label{eq: var}.\]

For theoretical purposes, we make the following assumption on the black-box function. 
\begin{assumption}
    \label{asm:kernel}
Assume the black-box function $f$ lies in the reproducible kernel Hilbert space (RKHS, $\mathcal{H}_k$) associated with kernel $k$. The RKHS norm of function $f$ is less than $B$ (i.e., $\|f\|_{\mathcal{H}_k} =\sqrt{\langle f,f\rangle_{\mathcal{H}_k}} \leq B$). We also, without loss of generality, assume the kernel function is normalized such that $k(\bm{x},\bm{x}) \leq  1,\forall \bm{x} \in M$. 
\end{assumption}

\subsection{Convergence of Algorithms}
A desired algorithm in this problem should have both the optimality gap and the movement cost converging to zero asymptotically. To capture this multi-objective nature, we evaluate the performance of an algorithm in terms of regret, which is the gap between the choice of the actual algorithm and the best choice of a clairvoyant. Let $f^*=\max_{\bm{x}\in M} f(\bm{x})$ be the ground truth of the optimal location $\bm{x}^*$, and define the instantaneous regret as the gap between the optimal location and a selected destination, i.e., $r_t = f^* - f(\bm{x}_t)$. 

One sufficient condition used widely \citep{srinivas2009gaussian} to prove asymptotic optimality is showing the sublinear convergence of the cumulative regret, which is defined as $R_T = \sum_{t=1}^T r_t$. To see this, note that if $R_T$ grows sublinearly, a subsequence must exist where the optimality gap $r_t$ converges to zero.

Similarly, we define the cumulative movement costs as $C_{T} = \sum_{t=1}^T c(\bm{x}_{t-1},\bm{x}_{t})$, where $\bm{x}_0$ is the initial location of the traveler. Our analysis tracks the two objectives separately to provide sufficient granularity for each. Such granularity can be beneficial when the two objectives cannot be directly summed up. Yet, if regrets and movement costs are measured on the same scale, it suffices to keep track of the cumulative loss, i.e.,
\(L_T = R_T + C_T.\)

To keep our analysis non-trivial, we assume the traveler can reach anywhere in the space $(M, c)$ from its origin with finite costs. Otherwise, we can always ignore the space that cannot be reached by the traveler within a finite budget. 
\begin{assumption}
\label{asm:physical}
     $(M, c)$ is a totally bounded and connected metric space that contains $\bm{x}_0$. For any two points $\bm{x}_i,\bm{x}_j\in M$, we can transport $\bm{x}_i$ to $\bm{x}_j$ with costs $c(\bm{x}_i,\bm{x}_j)<\infty$. 
\end{assumption}

\subsection{Batched Algorithms}
\label{sec:batch}
Usually, a BO algorithm chooses $\bm{x}_t$ using all the information in $\mathcal{D}_{t-1}$. That is, at each iteration, the decision at the next timestamp is given by a policy $\bm{x}_t = \pi(\mathcal{D}_{t-1})$ that takes current information $\mathcal{D}_{t-1}$ and outputs \textit{one} decision. For example, (single-point) Thompson Sampling (TS) can be treated as a policy 
\[\bm{x}_t=\pi_{\texttt{TS}}(\mathcal{D}_{t-1})=\arg\max(g), \quad g\sim \mathcal{GP}(\mathcal{D}_{t-1}).\]
However, sometimes (even if responses do not come in batches), it can be beneficial to make multiple decisions before observing the results. A batched algorithm picks multiple destinations simultaneously before visiting them. In this case, the policy $\pi^{\mathcal{B}}$ takes $\mathcal{D}_{t-1}$ but outputs \textbf{multiple} points. As a result, instead of making one decision at every timestamp, one now makes a batch of decisions at only $N(T)\leq T$ timestamps, where we use $N(T)$ to denote the number of batches. We denote the timestamp where batched decisions are made as $\mathcal{T}= \{t_0,t_1,\dots,t_{N(T)}\}$, with $1 = t_0 \leq t_1 < \dots < t_{N(T)} = T$. At every time $t_{i-1}$, a batched algorithm chooses the next $t_i-t_{i-1}$ destinations, denoted as $\mathcal{S}_i:=\{\bm{x}_{t_{i-1}},\dots, \bm{x}_{t_{i}}\}$.

For example, if we were to plan a batch from $t_{i-1}$ to $t_{i}$ using the batched TS \citep{kandasamy2018parallelised}, the policy $\pi^\mathcal{B}_{\texttt{TS}}$ becomes
\begin{equation}
\label{eq:ts}
    \mathcal{S}_i = \pi^{\mathcal{B}}_{\texttt{TS}}(\mathcal{D}_{t_{i-1}-1})=  \{\arg\max(g_\tau) :\tau\in t_{i-1}\leq \tau<t_i\}, \quad 
    g_\tau \sim \mathcal{GP}(\mathcal{D}_{t_{i-1}-1}).
\end{equation}

\subsection{Asymptotic Notations}
We use asymptotic notation to denote the dependency on the number of samples $T$. Among them, $O(\cdot)$ denotes the asymptotic upper bound for a specified algorithm. $\Omega(\cdot)$ denotes the asymptotic upper bound for all algorithms. $\Theta(\cdot)$ implies both $O(\cdot)$ and $\Omega(\cdot)$. $o(\cdot)$ denote asymptotically less. We also use $\sim$ to suppress any logarithmic dependencies on $T$ that do not involve dimension $d$.

\section{A Simple Framework: Plan More and Move Less}
\label{sec:method}
The proposed framework is inspired by an elegant theorem. In this section, we first present our framework in Sec.\ref{sec:alg}. Then, in Sec.\ref{sec:thm}, we introduce the main theorem and how it applies to our algorithm. Finally, we present the convergence result for three examples of batched algorithms.

\subsection{Plan-ahead with Batched Algorithms}

The core concept of our BO approach is straightforward: leverage any batched algorithm to identify a set of promising design candidates, then optimize the selection by taking a TSP tour through them. This process is outlined in Algorithm \ref{alg:plan}.

\label{sec:alg}
\begin{algorithm}[htb]
   \caption{Plan-ahead with Batched Algorithms}
   \label{alg:plan}
\begin{algorithmic}
   \State{\bfseries Input:} Metric space $(M,c)$, a batched algorithm $\pi^\mathcal{B}$
   \For{Batch index $i=1,\dots, N(T)$}
   \State Select a batch of designs $\mathcal{S}_i$ according to some user selected $\pi^\mathcal{B}$.
   \State Create a graph $\mathcal{G}_i$ using $\mathcal{S}_i$ and $c$.
   \State Plan a TSP path on $\mathcal{G}_i$;
   \State Observe designs in $\mathcal{S}_i$ following the TSP path;
   \EndFor
\end{algorithmic}
\end{algorithm}
Our framework seamlessly integrates with any existing batched algorithm. For any given batch algorithm $\pi^\mathcal{B}$, our method only adds one additional step by planning a TSP tour through the points in the batch (see Algorithm \ref{alg:plan}). Recall the notation from Sec.\ref{sec:batch} that $\mathcal{S}_i=\pi^\mathcal{B}(\mathcal{D}_{t_{i-1}-1})$. For each batch $i$, the traveler chooses a batch of points $\mathcal{S}_i$ according to some policy $\pi^\mathcal{B}$. Note that there is no restriction of $\pi^\mathcal{B}$ as long as it is a batched algorithm. Some valid examples of $\pi^\mathcal{B}$ include the batched TS \citep{kandasamy2018parallelised}, batched UCB \citep{chowdhury2019batch}, and batched pure exploration (BPE) algorithm \citep{li2022gaussian}. Again, we note that the use of a batched algorithm does not require responses to come in batches. With the selected batch $\mathcal{S}_i$, it suffices to consider only the space $M$ restricted to $\mathcal{S}_i$. This is essentially a graph $\mathcal{G}_i$ where the vertices are the unique destinations in $\mathcal{S}_i$, and the distances between vertices are defined by the given metric $c$. Then, the traveler finds a path through all the vertices in $\mathcal{G}_i$ by solving a TSP. After planning a TSP path, the traveler simply travels through the destinations following the TSP path. After visiting all destinations in $\mathcal{S}_i$, the traveler heads to the next iteration by selecting a new batch. 

\subsection{Practical Batched BO with Successive Elimination}
\label{apx:gpbase}

As shown above, our approach to reducing movement requires a batched algorithm. However, despite their theoretical success, it has been noted that non-batched algorithms could have faster convergence on practical problems compared to batched algorithms \citep{de2021greed}. This is intuitively understandable, as batched algorithms tend to prioritize diversity in the sampled batch, making them less greedy in terms of exploitation. In contrast, greedy and exploitative algorithms can converge more quickly in practice.

To enhance the practical performance of our algorithm, we incorporate the idea of successive elimination, which is a standard trick in bandit literature \citep{lattimore2020bandit}. The key idea of successive elimination is to avoid the exploration of hopeless designs. The successive elimination relies on a simple rule: no matter how the batched algorithm tries to promote exploration, it is meaningless to spend experiments on a point that is suboptimal with high probability.

Formally, let the lower and upper confidence bounds of the $\mathcal{GP}(\mathcal{D}_i)$ be $\texttt{LCB}_i(\cdot)$ and $\texttt{UCB}_i(\cdot)$ respectively, i.e., 
\begin{equation}
    \texttt{UCB}_t(\bm{x}) = \mu_{t}(\bm{x}) + \eta_t \sigma_{t}(\bm{x}),\quad 
    \texttt{LCB}_t(\bm{x}) = \mu_{t}(\bm{x}) - \eta_t \sigma_{t}(\bm{x}),
\end{equation}
where $\eta_t$ is a hyper-parameter that controls the level of exploration. Now, for any point $\bm{x}$, if there exist a point $\bm{x}'$ such that $\texttt{UCB}_t(\bm{x})<\texttt{LCB}_t(\bm{x}')$, then $\bm{x}$ cannot be the optimal solution with high probability. As such, we exclude them from our candidate set $\mathcal{X}_{i+1}$ for the next batch, which is updated with the following rule:
\begin{equation}
    \mathcal{X}_{i+1} = \{\bm{x}\in \mathcal{X}_{i}: \texttt{UCB}_{t_i}(\bm{x}) > \max_{\bm{x}} \texttt{LCB}_{t_i}(\bm{x})\}.
\end{equation}
In our experiments, $\eta_t$ is set to be $\beta_t/2$ to avoid hypermeter tuning in the algorithm, but users can choose them for practical purposes. 

\begin{algorithm}[htbp!]
   \caption{Batched Successive Elimination for $\mathcal{GP}$}
   \label{alg:BSE}
\begin{algorithmic}[1]
   \State{\bfseries Input:} Metric space $(M,c)$, kernel $k$, grid $\mathcal{T}$, policy $\pi^\mathcal{B}$
   \State $\mathcal{X}_1 \gets M$, $t\gets 1$, $\mathcal{D}_{0}\gets \emptyset$.
   \For{Batch $i=1,\dots,N(T)$}
   \State Select a batch of designs \textcolor{red}{$\mathcal{S}_i \subset \mathcal{X}_i$} according to some user select $\pi^\mathcal{B}$.
   \State Create a graph $\mathcal{G}_i$ using $\mathcal{S}_i$ and $c$.
   \State Plan a TSP path on $\mathcal{G}_i$;
   \State Observe designs in $\mathcal{S}_i$ following the TSP path;
   \textcolor{red}{
   \State Create model $\mathcal{GP}(\mathcal{D}_{t_i})$ for $\mu_{t_i}(\cdot)$ and $\sigma_{t_i}(\cdot)$;
   \State $\mathcal{X}_{i+1}\gets\{\bm{x}\in \mathcal{X}_i: \texttt{ucb}_{t_i}(\bm{x}) > \max_{\bm{x}} \texttt{lcb}_{t_i}(\bm{x}) \}$.
   }
   \EndFor
\end{algorithmic}
\end{algorithm}

The only difference between Algorithm \ref{alg:BSE} and \ref{alg:plan} is that Algorithm \ref{alg:BSE} takes designs only in the set $\mathcal{X}_i$, which has been labeled in purple in Algorithm \ref{alg:BSE}. To implement it, the only difference compared to Algorithm \ref{alg:plan} is that the feasible region for acquisition function optimization is now $\mathcal{X}_i$ rather than the whole space $M$. 

\section{Theoretical Results: The More You Plan, the Less You Travel}
\label{sec:thm}
The effectiveness of Algorithm \ref{alg:plan} was powered by one elegant property of shortest paths. It is known that for arbitrary $n$ points in a unit cube in the Euclidean space $\mathbb{R}^d$, there exists a path through the $n$ points whose length (measured in Euclidean length) does not exceed $O(n^{1-\frac{1}{d}})$ \citep{few1955shortest,beardwood1959shortest}. This interesting property implies that although we can easily construct examples where the path length grows linearly with the number of points, the shortest path through them is only sublinear. To see why this is true intuitively, consider the simplest case where the points lie in a unit interval $[0,1]$. No matter how many points one has, we can always visit them all by traveling from $0$ to $1$ with only a unit cost. 

Now, to utilize this idea for our purpose, we prove a generalized version in Theorem \ref{thm:general} such that it holds for any non-trivial metric spaces. For theoretical development, we give a rigorous definition of our dimension $d$. 

\subsection{Minkowski Dimension}
Minkowski Dimension is a generalization of dimensions in Euclidean spaces and is widely used to describe the complexity of a metric space \citep{bishop2017fractals}. We built our theoretical results on the Minkowski Dimension to align with previous works \citep{koren2017multi}.
\begin{definition}[Balls in Metric Spaces]
    For metric space $(M,c)$, a \textit{ball} centered at $\bm{z}$ with radius $\epsilon>0$ is a set of points satisfying
    \(\mathcal{B}_{\epsilon}(\bm{z}) = \{\bm{x}\in M: c(\bm{x},\bm{z}) < \epsilon\}.\)
\end{definition}

\begin{definition}[Minkowski Dimension by Covering]
\label{def:dim}
    For a set of points $Z \subset M$, we define $N(\epsilon)\triangleq \inf |Z|$ satisfying
    \(M\subseteq \bigcup_{\bm{z}_i\in Z} \mathcal{B}_{\epsilon}(\bm{z}_i).\)
    We define the \textit{dimension} of $(M,c)$ as 
    \[d \triangleq \lim_{\epsilon\to 0}\frac{\ln(N(\epsilon))}{\ln(1/\epsilon)}.\]
\end{definition}

\begin{definition}[Packing Number]
\label{def:pack}
    For metric space $(M,c)$, we define the packing number of $M$ as $N^\texttt{p}(\epsilon)\triangleq \sup_{Z\subseteq \mathcal{Z}} |Z|$ where 
    \(\mathcal{Z} = \{Z\mid \forall \bm{z}_i, \bm{z}_j \in Z,\,\mathcal{B}_{\epsilon}(\bm{z}_i) \cap \mathcal{B}_{\epsilon}(\bm{z}_j) = \emptyset,\,\bigcup_{\bm{z}_i\in Z} \mathcal{B}_{\epsilon}(\bm{z}_i) \subseteq M\}.\)
\end{definition}
\begin{remark}
\label{rmk:packcover}
    Covering and packing are two alternative and closely related definitions for the Minkowski dimension. It is known that
    \(N^\texttt{p}(\epsilon) \leq N(\epsilon) \leq N^\texttt{p}(\epsilon/4)\) \citep{bishop2017fractals}.
\end{remark}

\begin{theorem}
    \label{thm:general}
    Fix any metric space $(M,c)$ of dimension $d$ satisfying Assumption \ref{asm:physical}. For any set of $n$ points in $M$, we can find a path through all of them whose length is at most 
    \(O(n^{1-\frac{1}{d}}).\)
    This bound is tight in the sense that for any such metric space, there always exists $n$ points such that the shortest path connecting them is at least \(\Omega(n^{1-\frac{1}{d}}).\)
\end{theorem}
\begin{proof}
    This proof is by construction.

    By Definition \ref{def:dim}, we can always find $n$ open balls with radius $\epsilon \leq O(n^{-\frac{1}{d}})$ such that they cover the whole space $M$. We find a covering that is minimal, so any $n-1$ balls cannot cover the whole space. 

    To construct a required path, we will first construct a path through all the centers of the balls and connect the destinations $\bm{x}_i$ to the path. We will then show the length of the resulting path is less than $O(n^{1-\frac{1}{d}})$.

    We start by constructing a path through all the centers of the balls $\{\bm{z}_1, \dots, \bm{z}_n\}$. Because we assume the metric space is connected, $\{\bm{z}_1, \dots, \bm{z}_n\}$ and the metric $c$ jointly defined a graph whose vertices are those centers and the edge weights are the lengths between them. One way to find a path through them is to use the minimum spanning tree on the graph. One can use a greedy algorithm, known as Prim's algorithm \citep{prim1957shortest}, to construct a minimum spanning tree of the $n$ centers. Prim's algorithm starts with an arbitrary vertex in the graph and grows a tree by iteratively adding the next vertex with the minimum edge weight (minimum length) to the current tree. Denote the tree at iteration $i$ being $\mathcal{T}_\tau$. It is guaranteed that $\mathcal{T}_n$ is a minimum spanning tree.
    
    We claim that at any iteration of Prim's algorithm, as long as the algorithm is not terminated, there always exists a point $\bm{z}_i\in \mathcal{T}_\tau$ in the current tree and a remaining point $\bm{z}_j\notin \mathcal{T}_\tau$ such that we can find a path from $\bm{z}_i$ to $\bm{z}_j$ with length less than $2\epsilon$. We prove this by construction.
    
    We let $\mathcal{B}_\tau = \bigcup_{\bm{z}_i\in \mathcal{T}_\tau}\mathcal{B}_{\epsilon}(\bm{z}_i)$ be the space covered by the open balls whose centers are already in the tree $\mathcal{T}_\tau$. For any $\mathcal{B}_\tau$, as long as the algorithm is not terminated, $M$ is not covered by $\mathcal{B}_\tau$ because we require the open ball cover to be minimal. Moreover, because $M$ is connected and is not entirely covered by $\mathcal{B}_\tau$, we can always find a point on the boundary of $\mathcal{B}_\tau$. In other words, we can find $a\in M\setminus \mathcal{B}_\tau$ and an open ball  $\mathcal{B}_{\epsilon}(\bm{z}_i)\in \mathcal{B}$ such that $c(\bm{z}_i,a)=\epsilon$.  Since $a$ is not covered by $\mathcal{B}_\tau$ but in $M$, there must exist another open ball $\mathcal{B}_{\epsilon}(\bm{z}_j)$ that covers it, so we have $c(\bm{z}_j,a)<\epsilon$. By the construction of the open balls, we have $\bm{z}_i,\bm{z}_j\in M$. Therefore, $\bm{z}_i \to a \to \bm{z}_j$ is a path contained in $M$ that has a length less than $2\epsilon$.
    
    Therefore, at every step of Prim's Algorithm, the weight (length) of the tree does not increase more than $2\epsilon$, so the total weight (length) of the minimum spanning tree is less than $2\epsilon n$. It is well-known that the TSP path is upper bounded by two times the length of a minimum spanning tree \citep{laporte1992traveling}. Hence, we conclude that the TSP path through the $n$ centers is at most $4 n \epsilon = O(n^{1-\frac{1}{d}})$.

    Now what remains is to connect $\{\bm{x}_1, \dots, \bm{x}_n\}$ to the existing path. Because the open balls jointly cover the whole space, we know that any point $\bm{x}_i$ will fall into some ball, which means there exists a ball $\mathcal{B}_{\epsilon}(\bm{z}_i)$ such that $c(\bm{x}_i,\bm{z}_i) < \epsilon$.  Hence, we can simply add $\bm{x}_i$ by adding path $\bm{z}_i\to \bm{x}_i \to \bm{z}_i$. Recall that by the choice of $\bm{z}_i$, we have $c(\bm{x}_i,\bm{z}_i) < \epsilon$, so the length of path $\bm{z}_i\to \bm{x}_i \to \bm{z}_i$ is less than $2\epsilon$. Because we have $n$ such paths to add, the length of the resulting path is less than $2n\epsilon$.

    As such, we conclude that the constructed path has a length of at most \[4n\epsilon+2n\epsilon \leq 6 n \epsilon \leq O(n^{1-\frac{1}{d}}).\] 

    To show that there always exists $n$ points such that the shortest path connecting them is at least \(\Omega(n^{1-\frac{1}{d}})\), instead of using the covering number, we consider the packing number. By remark \ref{rmk:packcover}, the aforementioned space can pack at most $n$ open balls with radius $\epsilon/4$. We use the same trick to construct a minimum spanning tree, and its length should be lower bounded by $n\epsilon/2$. Because we know the length of a minimum spanning tree lower bounds the length of any TSP path \citep{laporte1992traveling}, we now have constructed a set of points with a shortest path of \(\Omega(n^{1-\frac{1}{d}}).\)
\end{proof}
\begin{remark}
As mentioned in Sec.\ref{sec:lit}, although solving the exact shortest TSP is NP-hard, there are polynomial-time algorithms that can find a constant approximation of the optimal solution. Because the error of the constant approximation is only a fixed constant multiplier, the guarantee in Theorem \ref{thm:general} holds the same for the actual path we find.
\end{remark}

Theorem \ref{thm:general} implies that the average distance through $n$ points decays at a rate of $O(n^{-\frac{1}{d}})$. To translate the result in our problem: if we can plan the future destinations in advance, \textit{the more future destinations you plan, the less distance you travel on average.} Although determining all the actions in one shot is unrealistic and can lead to huge regret, there are many effective batched algorithms that require only $O(\log\log T)$ batches \citep{gao2019batched,feng2022lipschitz,li2022gaussian} to achieve sublinear regret. This is highlighted in the following sections.

\subsection{Convergence Guarantees}
\label{sec:converge}

It is crucial to realize that Algorithm \ref{alg:plan} tests the exact same set of points chosen by the user's selected batched algorithm. The only difference lies in the order in which they are tested, as dictated by the TSP solution. Consequently, it retains the same theoretical guarantees on cumulative regret as the chosen batched algorithm. Therefore, as long as the batched algorithm converges in cumulative regret, we only need to show that the movement costs are also convergent. This section presents some desired theoretical results that follow directly from Theorem \ref{thm:general}.

We start with a general result that any batched algorithm that takes $o(T^{\frac{1}{d}})$ batches enjoys sublinear cumulative movement costs.
\begin{corollary}
    \label{thm:general_move_costs}
    Fix any metric space $(M,c)$ of dimension $d$ satisfying Assumption \ref{asm:physical}. For any batched algorithms with $N(T) = o(T^{\frac{1}{d}})$, following the TSP path for every batch, the cumulative movement costs $C_T$ is at most $o(T)$.
\end{corollary}
\begin{proof}
    This is a direct result of Theorem \ref{thm:general}. The total movement costs are upper-bounded by the number of batches multiplied by the costs of each batch. Given the costs of each batch cannot exceed the shortest path through all the $T$ points (i.e., $o(T^{1-\frac{1}{d}})$), we have the \(C_T = o(T).\)
\end{proof}

Note that Theorem \ref{thm:general} only contributes to the cumulative movement costs $C_T$. Because our method can use any given policy (as shown in Algorithm \ref{alg:plan} line $3$), it inherits exactly the same regret bound on $R_T$ as the given policy. As such, we give three examples of how our result can be adapted to existing batched algorithms. 

\subsubsection{Theory for BO}
To prove the convergence of our algorithm in BO, it suffices to find an algorithm that takes $o(T^{\frac{1}{d}})$ batches and has a sublinear regret. One of the candidates is the Batched Pure Exploration (BPE) algorithm proposed by \cite{li2022gaussian}. 

\begin{theorem}[\citealp{li2022gaussian}]
\label{thm:BPE}
Under assumption \ref{asm:kernel}, there exists an algorithm  that uses $N(T)= O(\log\log T)$ batches, and with high probability, the cumulative regret $R_T$ is at most $\tilde{O}(\sqrt{\gamma_T T})$
\begin{itemize}
    \item For RBF kernel, 
    \(\gamma_T = \tilde{O}\left(\log(T)^d\right)\)  
    \item For Mat\'{e}rn kernel,
    \(\gamma_T = \tilde{O}\left(T^{\frac{d}{2\nu+d}}\right),\) 
    where $\nu$ is a hyper-parameter in kernel.
\end{itemize}
\end{theorem}

Because the algorithm in Theorem \ref{thm:BPE} uses only $O(\log\log T)$ batches, and by our Theorem \ref{thm:general} every batch incurs at most $O(T^{1-\frac{1}{d}})$ movement costs, if we follow the TSP to conduct all the experiments in these batches, the total movement cost is at most $O(\log\log T \cdot T^{1-\frac{1}{d}})$. Hence, by choosing the $\pi^\mathcal{B}$ in Algorithm \ref{alg:plan} as BPE, we have convergence on both movement costs and regret.

\begin{corollary}
\label{thm:BO}
Under assumptions \ref{asm:kernel} and \ref{asm:physical}, there exists an algorithm whose cumulative loss $L_T$ is at most $\tilde{O}(T^{1-\frac{1}{d}}+\sqrt{\gamma_T T})$.
\end{corollary}

Corollary \ref{thm:BO}, combines both $R_T$ and $C_T$ into one unified result.

\subsubsection{Theory for other Bandit Settings}
A central property of Corollary \ref{thm:general_move_costs} is that it holds for any batched algorithm in any setting, as long as we are in a metric space. Here, we show two examples of using our algorithm for the multi-armed and Lipchitz bandit problem.

\paragraph{Stochastic $K$-armed Bandit}

Recently, \citet{gao2019batched} provided a batched successive elimination algorithm that can achieve optimal regret of $\tilde{O}(\sqrt{KT})$ with $O(\log\log T)$ batches. For the MAB problem, the metric space $M$ is a set of discrete points, so the TSP tour passing an arbitrary number of points in the metric space is upper bounded by a constant. In other words, the movement cost of any batch is $O(1)$, giving a total movement cost of $O(\log\log T)$.

\paragraph{Stochastic Lipschitz Bandit}

Following the success of batched MAB, \citet{feng2022lipschitz} uses the idea of successive elimination in Lipschitz bandits to achieve an optimal regret of $\tilde{O}(T^{\frac{d+1}{d+2}})$ with $O(\log\log T)$ batches. Similarly, when combined with our algorithm, movement costs are at most $\tilde{O}(T^{1-\frac{1}{d}})$.

\section{Numerical experiments}
\label{sec:experiment}
In this section, we compare our methods with existing methods. Our benchmarks are:
\begin{itemize}
    \item \texttt{TS}: naive TS, movement-unaware BO approach
    \item \texttt{UCB}: naive UCB, movement-unaware BO approach
    \item \texttt{TUCB}: our method with batched UCB as the batched policy
    \item \texttt{TTS}: our method with batched TS as the batched policy
    \item \texttt{SnAKe} \citep{folch2022snake}
    \item \texttt{GS} \citep{shi2023bayesian}
    \item \texttt{MINI} \citep{calandriello2022scaling}.
\end{itemize}
The experiments are conducted using BoTorch \citep{balandat2020botorch}. We tested all the methods on five synthetic test functions. Table \ref{tab:setup} contains detailed experimental setups for them.
\begin{table}[b]
\centering
\scriptsize
\begin{tabular}{ccclc}
\hline
Function & $d$ & $M$                    & $c(a,b)$    & $\epsilon$           \\ \hline
Ackley   & $2$ & $[-32.768,32.768]^2$   & $\|a-b\|_2$ & $\mathcal{N}(0,1.0)$   \\
Branin   & $2$ & $[-5,10]\times [0,15]$ & $\|a-b\|_2$ & $\mathcal{N}(0,3.0)$   \\
DropWave   & $2$ & $[-5.12,5.12]^2$   & $\|a-b\|_2$ & $\mathcal{N}(0,0.01)$   \\
Grewick  & $2$ & $[-20,20]^2$           & $\|a-b\|_2$ & $\mathcal{N}(0,0.01)$ \\
Levy     & $6$ & $[-5,5]^6$           & $\|a-b\|_2$ & $\mathcal{N}(0,1)$   \\ \hline
\end{tabular}
\caption{Experimental Setups}
\label{tab:setup}
\end{table} Because hyperparameter tuning is pathological in black-box function optimization, we set the hyperparameters following the rule of thumb: we set the exploration factor $\beta_t$ for UCB-like algorithms to be $4$ (effectively $2$ standard deviations) and correspondingly set the elimination confidence $\eta_t = 1$ (effectively $1$ standard deviation). As mentioned in Sec.\ref{sec:thm}, the only requirement for the batch selection is that the batch size should be $o(T^{\frac{1}{d}})$. To be consistent with our theory, we let the batch size grow exponentially by a factor of $c=1.1$ so that the total number of batch sizes is of order $O(\log T)$. We note that our selection of batch sizes is just a simple proof of concept, and it can be flexible depending on the application. Practically speaking, the faster the batch size grows, the fewer batches and the fewer movement costs. As a trade-off, the fewer the batches, the slower the algorithm adapts to newly acquired information, which could potentially slow down its optimization process. That said, any choice of $c>1$ enjoys the same convergence guarantees (up to some constant). 

The results are reported in Fig. \ref{fig:tests} and \ref{fig:tests2}. We run every method on every function for $T=100$ rounds (including the initial point) for $5$ replications.  To avoid confusion, we note that all the algorithms are tested under non-batched settings where algorithms can observe the outcomes immediately after the query. For better visualization, for any time step $T$, we report the average regret of the last $T/2$ steps, (e.g., $\lceil T/2\rceil\sum_{1=t}^{\lceil T/2\rceil}r_t$ for regret). We do the same for movement costs, where at any $T$ we report $\lceil T/2\rceil\sum_{t=1}^{\lceil T/2 \rceil}c(\bm{x}_{t-1}, \bm{x}_{t})$, and $c(\bm{x}_{t-1}, \bm{x}_{t})$ is the absolute Euclidean distance. The solid line represents the mean of the $5$ replications, and the upper/lower bounds of the shaded region correspond to the maximum and the minimum of the five replications.

Note that movement costs and regret are reported separately due to the multi-objective nature of our problem. Of course, one can take a weighted sum of each, but since their scales Reporting them separately allows for flexibility in interpretation and the possibility of adjusting their relative importance according to the context

Based on the figures, we find that our method has superior performance across different test cases by achieving comparable regret as naive baselines while maintaining smaller movement costs. Below, we highlight some interesting insights from the results. 

First, aligned with our analysis in Sec.\ref{sec:intro}, penalization methods like \texttt{GS} tend to be myopic in balancing between exploration and movement costs, which limits their exploration of different locations. In some cases, \texttt{GS} hardly moves in the course of optimization, which incurs almost constant regret. Although \texttt{GS} works very well for the Branin function, this is due to the fact that all the local optimal solutions for Branin are indeed global optimal solutions. Therefore, getting sucked at a local optimal solution becomes a blessing rather than a curse for it.

Second, we observe very limited evidence of \texttt{SnAKe}'s convergence. Although it does not get stuck like \texttt{GS}, it has difficulties converging on regret and movement costs as the optimization proceeds. We suspect this is a side effect of its heuristic $\epsilon$-point deletion, which helps escape local optima by preventing revisiting regions that have been visited before. While $\epsilon$-point deletion helps the algorithm avoid local optima, it also prevents it from staying at the global optimum. Furthermore, despite the fact that \texttt{SnAKe} plans routes for the future, it constantly changes its plan after receiving new observations. This forfeits the benefits of path planning outlined in Theorem \ref{thm:general}. In contrast, our algorithm does not change the planned destination until the next batch. This difference in results again emphasizes the importance of planning ahead.

Third, although \texttt{MINI} theoretically has the lowest movement cost, its approach of evaluating the same few candidates makes it practically inefficient in understanding the landscape of the black-box function. As a result, its optimization process is almost always slower compared to both \texttt{UCB} and \texttt{TUCB}.

Fourth, we observe that \texttt{UCB} has low movement costs on certain test functions, such as Ackley. This aligns with the discussion in Sec.\ref{sec:intro}, where algorithms that do not specifically account for movement can still exhibit decaying movement costs if the black-box function has a clear global optimum. In such cases, pursuing better designs naturally results in reduced movement costs.

Finally, \texttt{UCB} outperforms \texttt{TS}, which is expected given the exploratory nature of \texttt{TS}, leading to excessive movement and consistently higher movement costs, as it is less greedy. However, when paired with our method, the movement costs of \texttt{TTS} decrease significantly. This not only validates our algorithm's effectiveness but also highlights its adaptability to any acquisition function chosen.

\begin{figure}[H]
\centering
\includegraphics[width=0.49\linewidth]{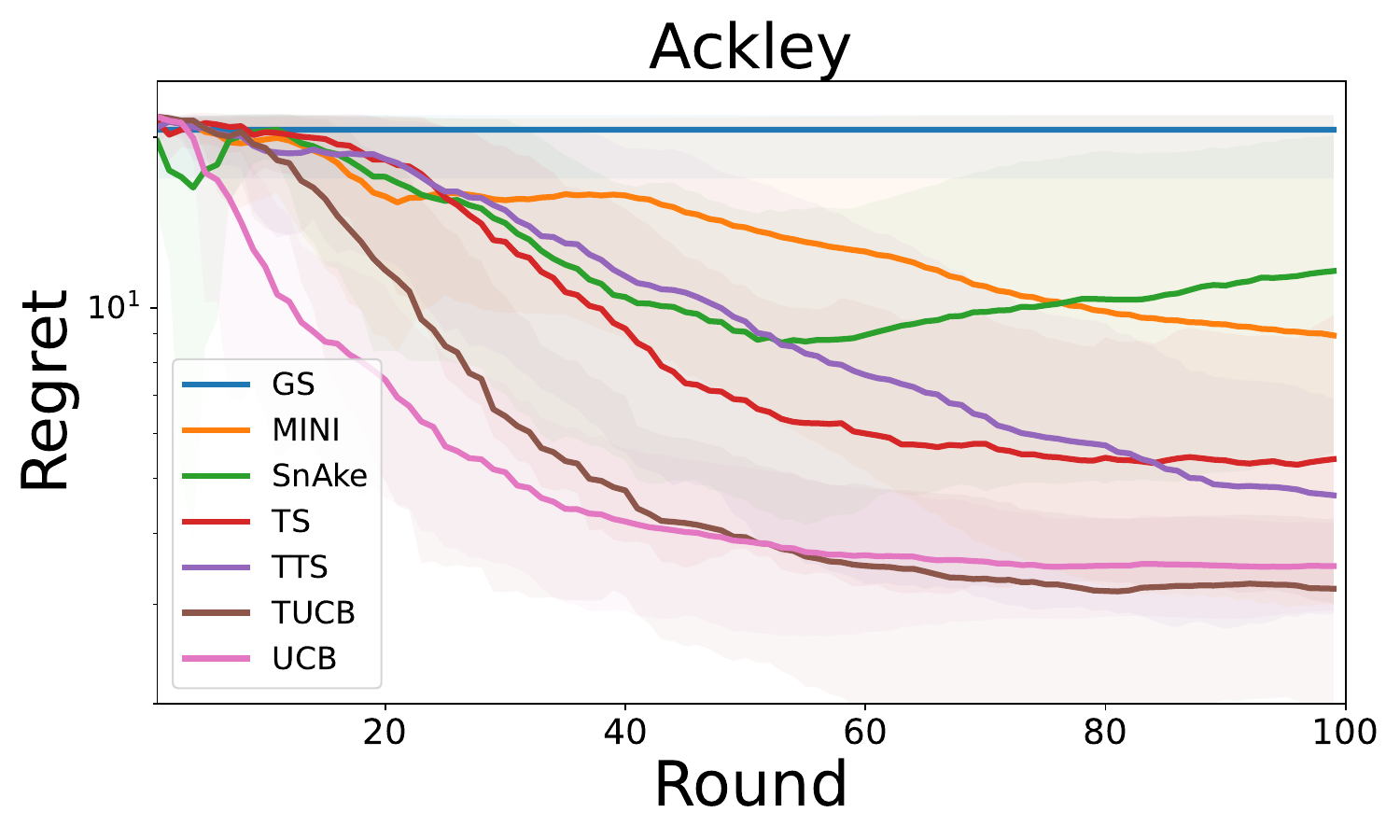}
\hfill
\includegraphics[width=0.49\linewidth]{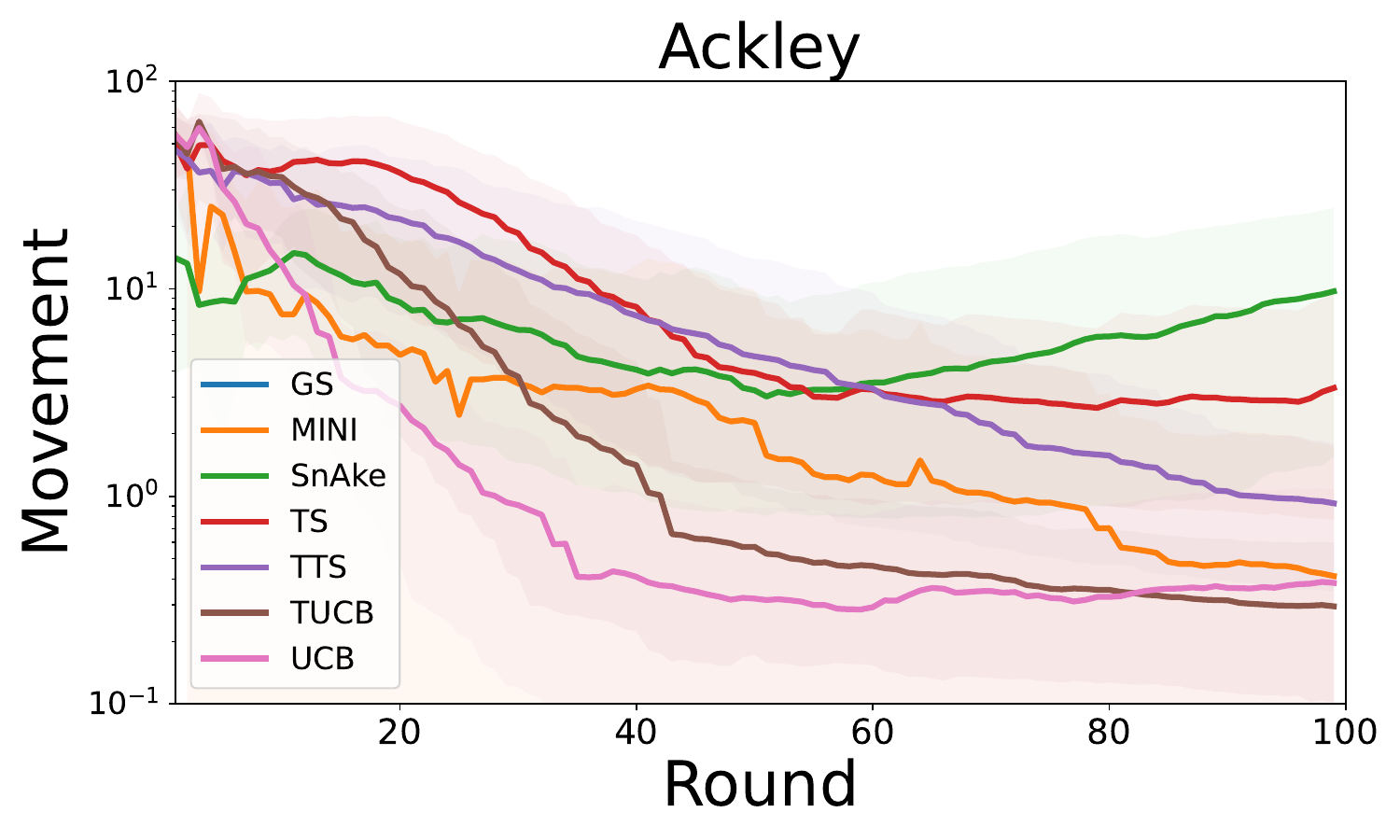}
\includegraphics[width=0.49\linewidth]{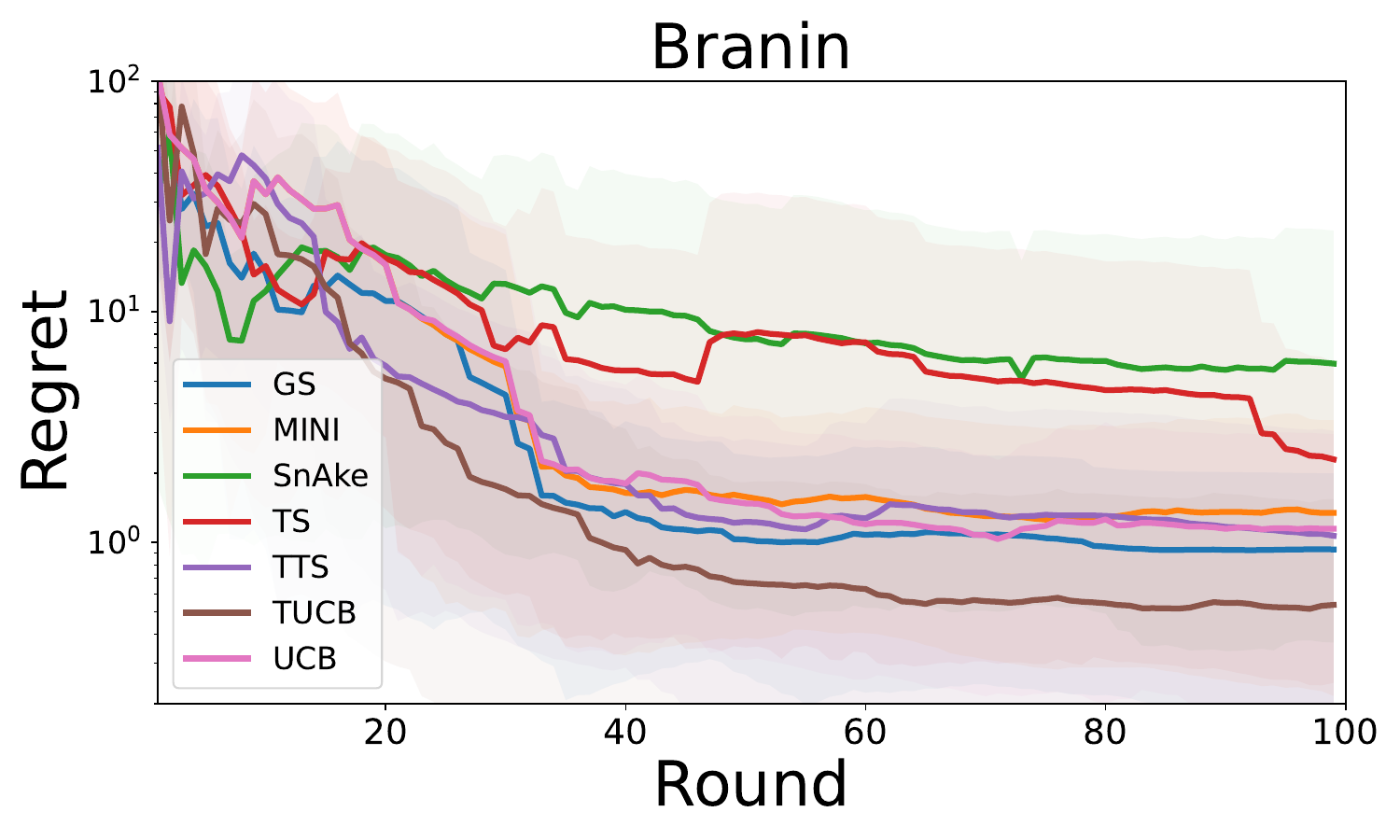}
\hfill
\includegraphics[width=0.49\linewidth]{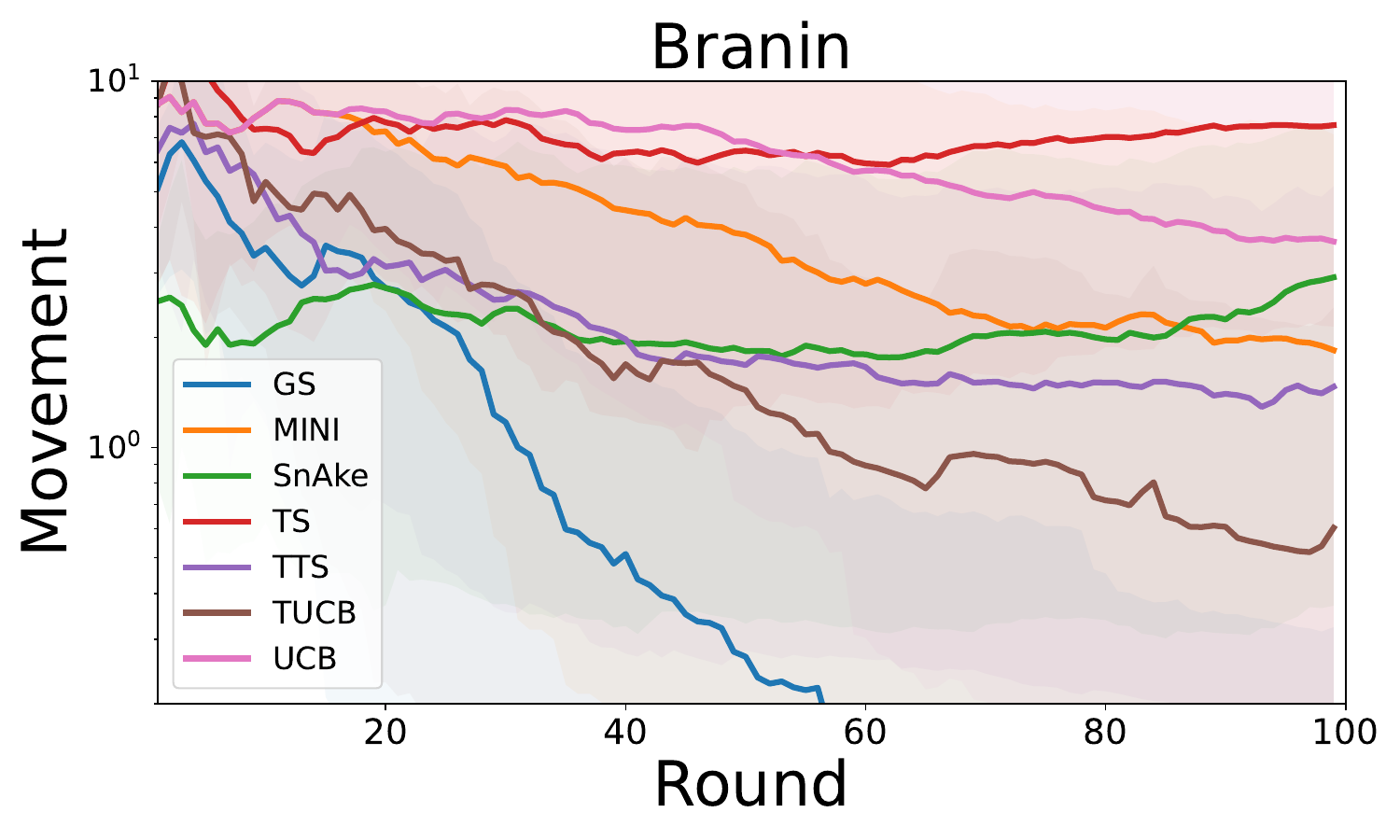}
\caption{Regret (left) and Movement costs (right) on Synthetic Test Benchmarks}
\label{fig:tests}
\end{figure}
\begin{figure}[H]
\centering
\includegraphics[width=0.49\linewidth]{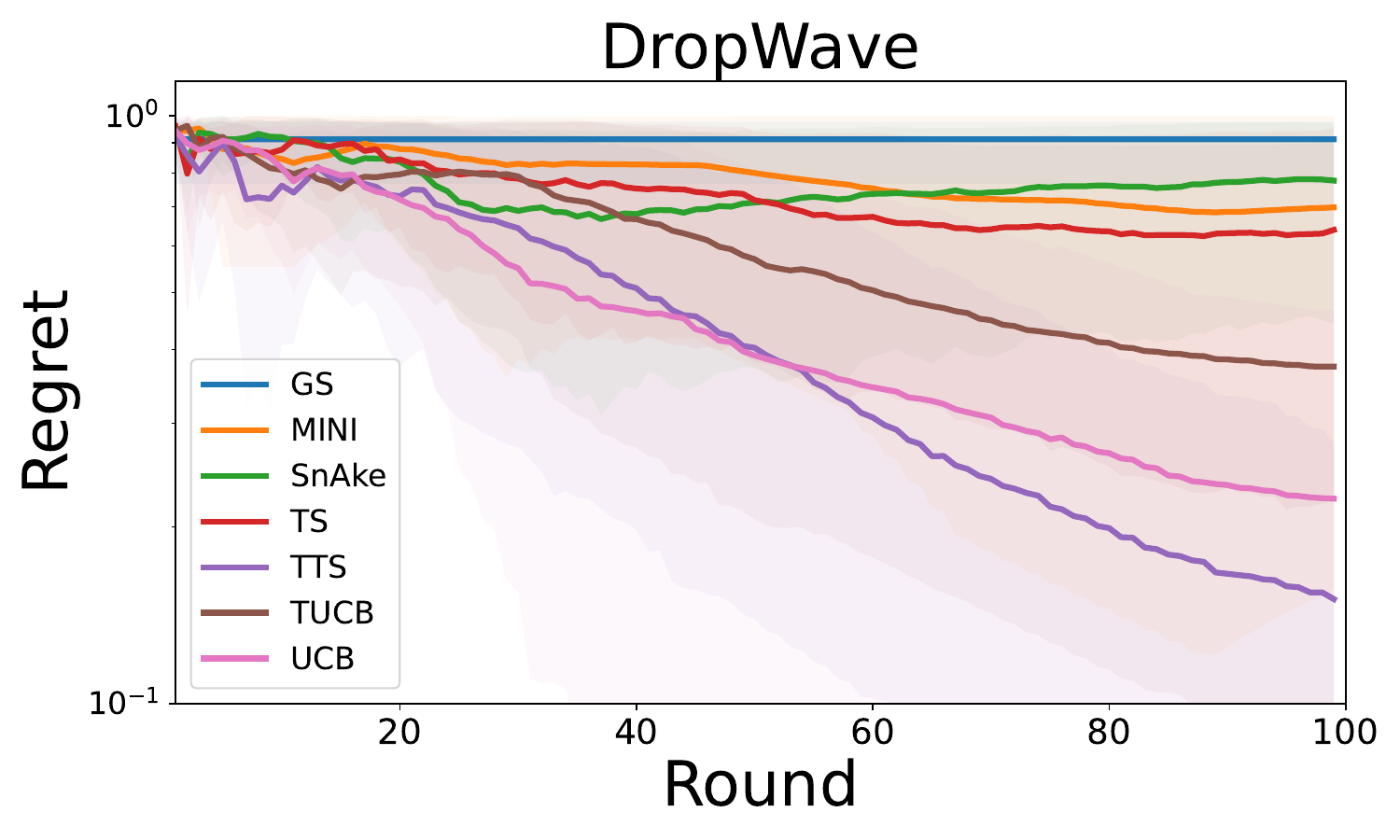}
\hfill
\includegraphics[width=0.49\linewidth]{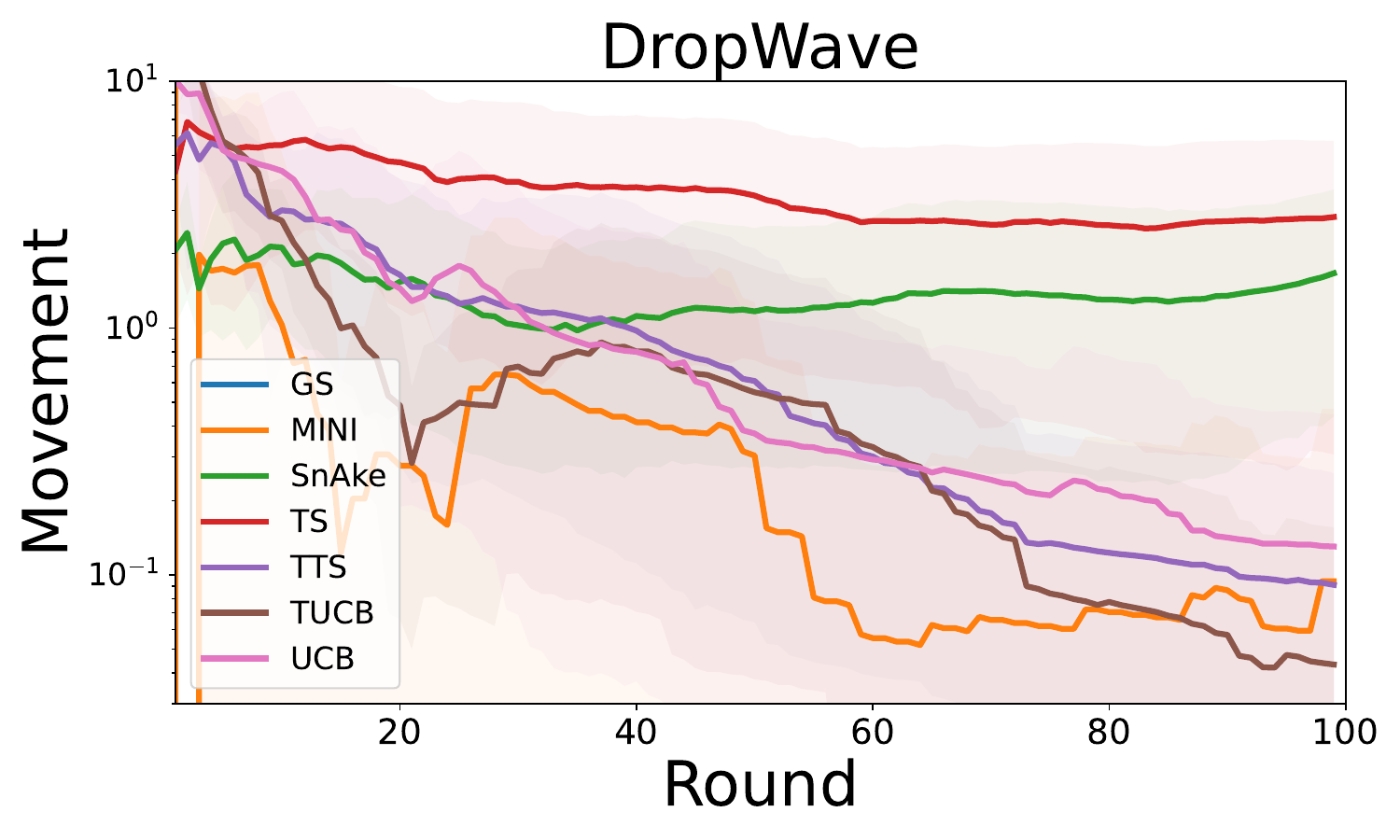}
\includegraphics[width=0.49\linewidth]
{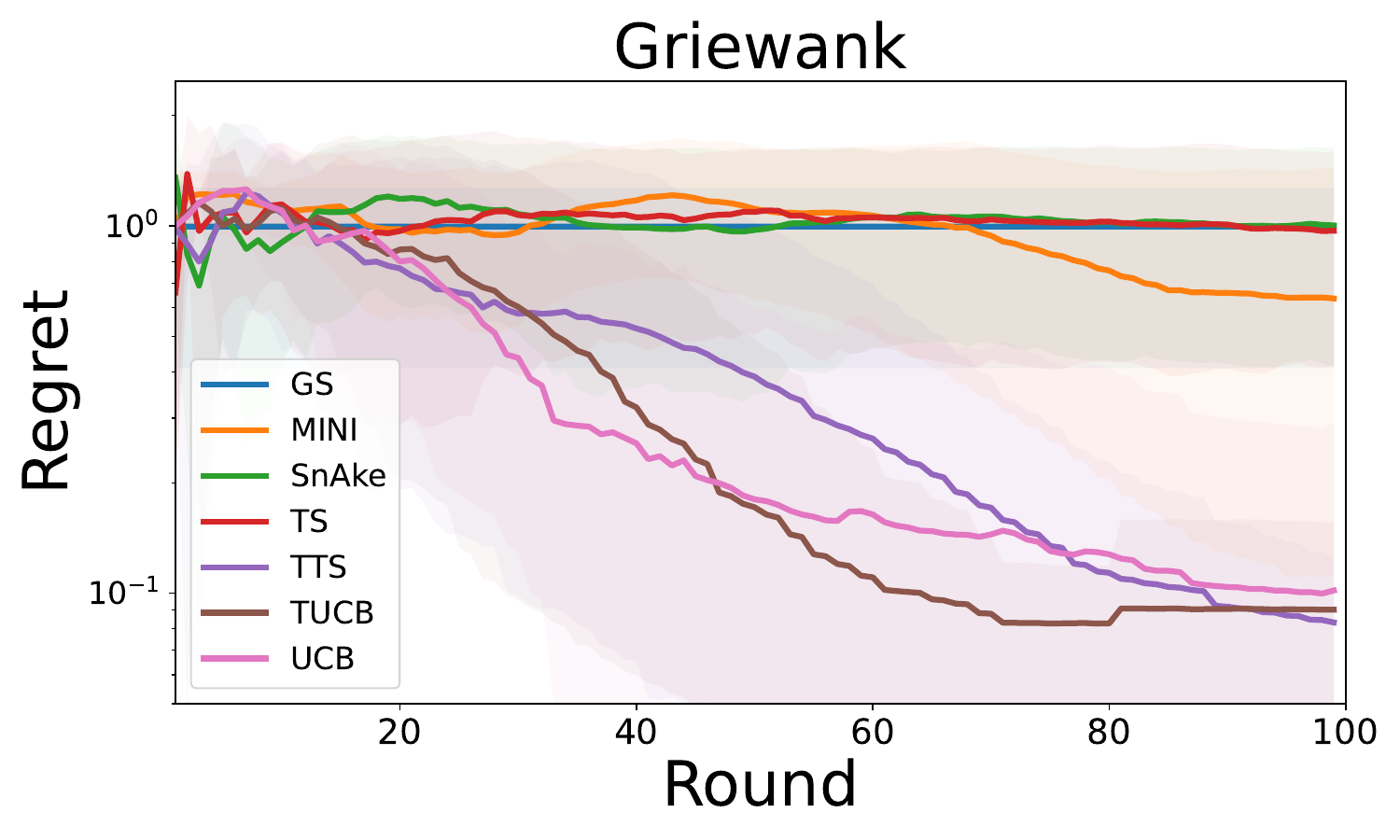}
\hfill
\includegraphics[width=0.49\linewidth]{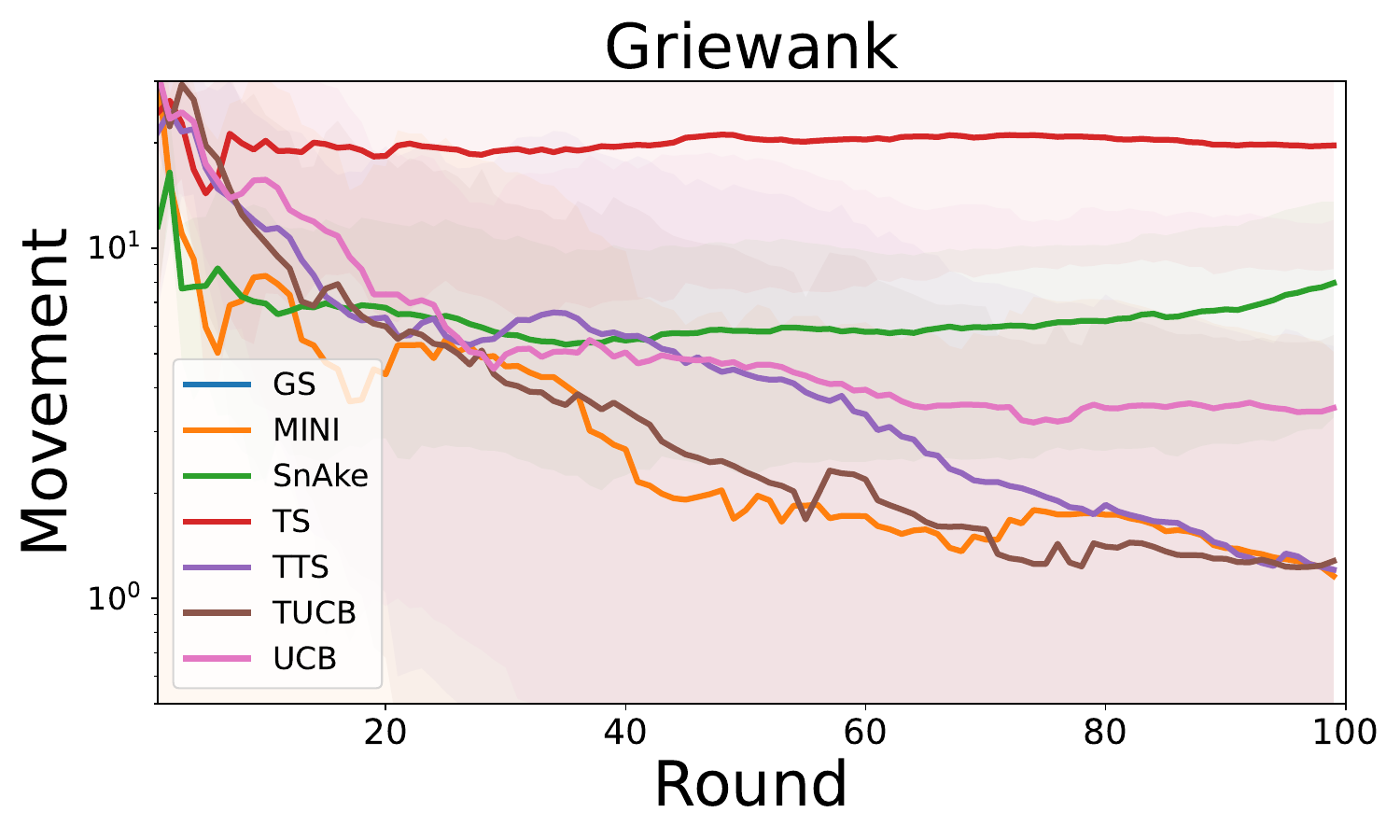}
\includegraphics[width=0.49\linewidth]{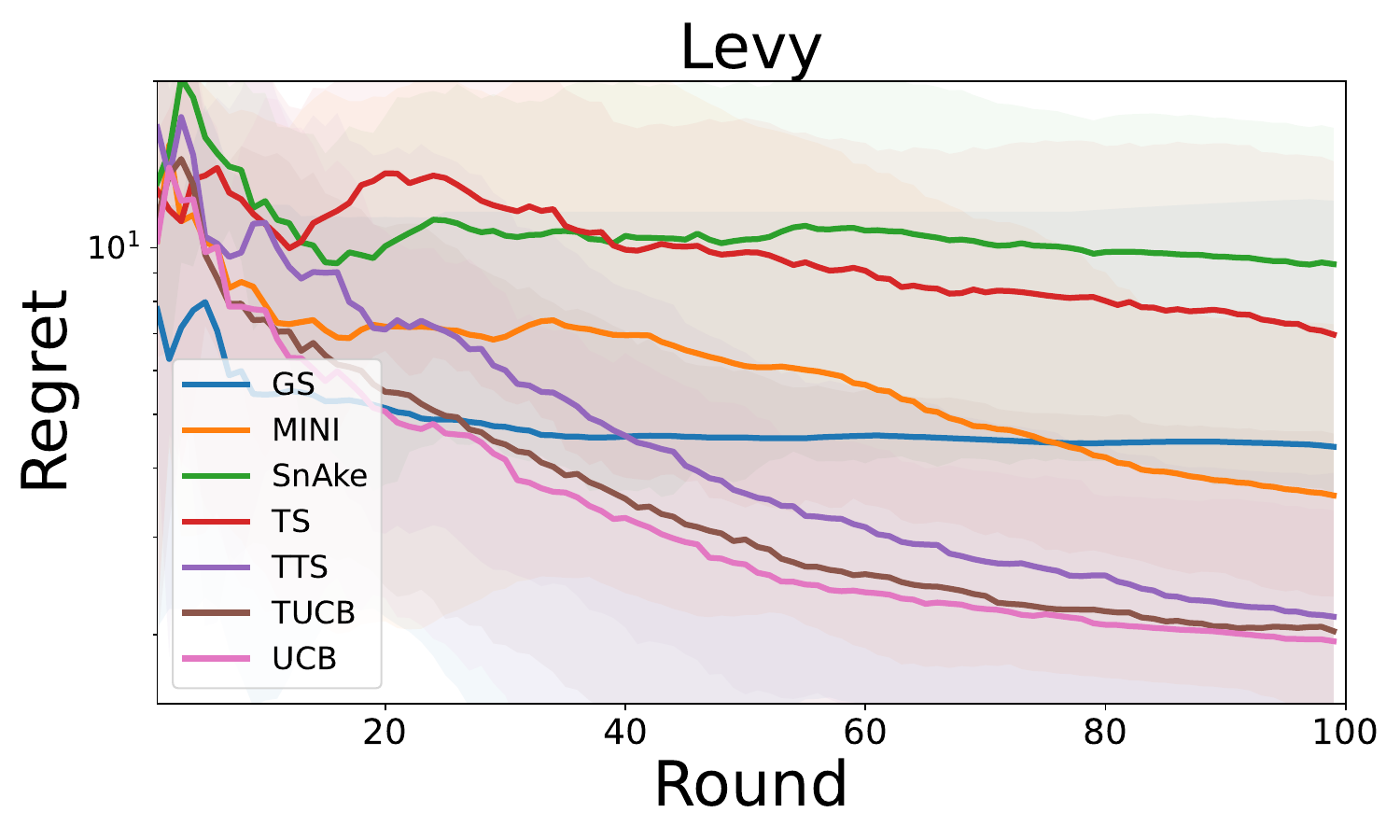}
\hfill
\includegraphics[width=0.49\linewidth]{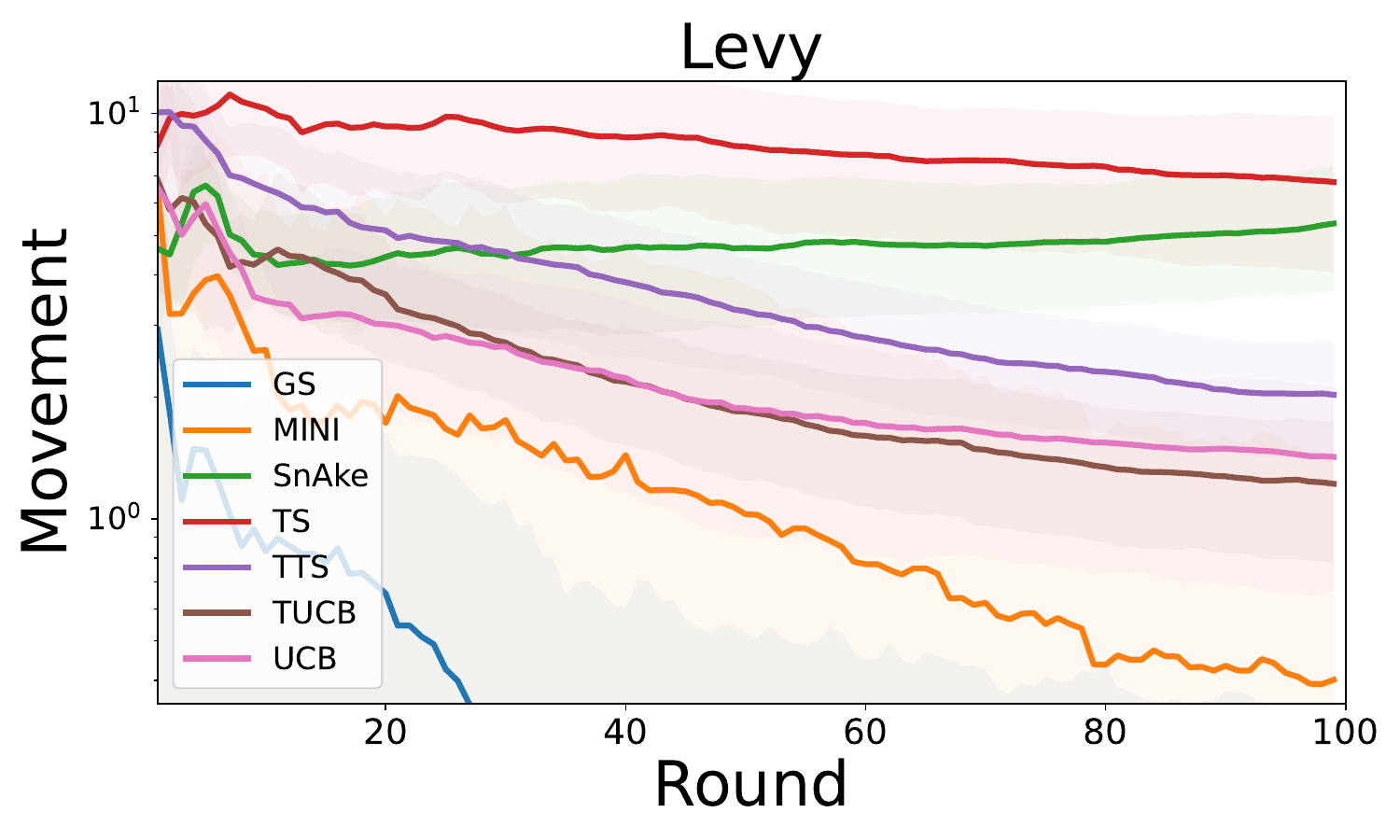}
\caption{Regret (left) and Movement costs (right) on Synthetic Test Benchmarks}
\label{fig:tests2}
\end{figure}

\section{Conclusion and Discussion of Future Work}
This paper introduces a framework for BO with movement costs by taking TSP tours within batched algorithms. Our approach is a plug-in that can be seamlessly integrated into any batched algorithm, whether in BO or general bandit algorithms. We demonstrate that our method offers both theoretical convergence and superior practical performance. We conclude by noting that, although our method achieves optimal rates, the growth rate of movement costs may be faster than that of regrets when the metric space is in high dimensions. Future work could focus on developing algorithms with less movement dependency on dimensionality.

\bibliography{references}

\end{document}